\documentclass[a4paper,12pt]{article} \pdfoutput=1

\usepackage{amsmath,amssymb,amsfonts,amsthm}
\usepackage{todonotes}
\usepackage[caption=false]{subfig}
\usepackage{microtype}
\usepackage{hyperref}
\hypersetup{pdfauthor={Lisa Glaser \& Abel B. Stern},pdftitle={Reconstructing manifolds from truncated spectral triples}}

\usepackage{grffile}

\usepackage{braket}

\usepackage{algpseudocode}

\newcommand{\Z}{\mathbb{Z}}
\newcommand{\R}{\mathbb{R}}
\newcommand{\C}{\mathbb{C}}

\newcommand{\eqdef}{\overset{\underset{\mathrm{def}}{}}{=}}
\newcommand{\norm}[1]{ \left\| #1 \right\| }

\newcommand{\DS}{D_{S^2}}
\newcommand{\DB}{D_{S^2} + c B}
\newcommand{\fdspace}{\mathcal{D}}
\newcommand{\cstar}{C$^*$}
\newcommand{\disp}{\eta}

\newcommand{\mstatespace}[1]{\mathbb{P}(#1)}
\newcommand{\statespace}{\mstatespace{H_\Lambda}}
\newcommand{\estatespace}[1]{\mstatespace{H_\Lambda}_{#1}}
\newcommand{\statemap}{\Phi_\Lambda}

\newcommand{\contdist}{d_{M}}
\newcommand{\lambdist}{d_{\Lambda}}
\newcommand{\trlambdist}{\widetilde{d_{\Lambda}}}
\newcommand{\rndist}[2]{d_{\R^n}\left(#1, #2\right)}
\newcommand{\embedding}{\phi}
\newcommand{\brkt}[1]{\mathbf{#1}}

\newcommand{\projective}[1]{\mathbb{P}(#1)}

\newcommand{\lip}[3]{\operatorname{Lip}_{#3}^{(#1)} (#2)}
\newcommand{\ip}[1]{\left\langle #1 \right\rangle}
\newcommand{\spip}[1]{\left( #1 \right)_{\spinorbundle}}
\newcommand{\tbigo}[1]{\widetilde{O} \left( #1 \right)}
\newcommand{\bigo}[1]{O \left( #1 \right)}
\newcommand{\spinc}{spin$^c$}

\newcommand{\hk}[3]{h^{#1}_{#2 #3}}
\newcommand{\pk}[3]{p^{#1}_{#2 #3}}
\newcommand{\plift}[2]{p^{#1}(#2)}
\newcommand{\pliftmap}[1]{p^{#1}}
\newcommand{\spectralkernel}[1]{\widetilde{E}_{#1}}
\newcommand{\spectrallift}[1]{E_{#1}}

\DeclareMathOperator{\id}{id}

\DeclareMathOperator{\vol}{vol}
\DeclareMathOperator{\End}{End}
\DeclareMathOperator{\rank}{rank}

\theoremstyle{plain}
\newtheorem{theorem}{Theorem}[subsection]
\newtheorem{lemma}[theorem]{Lemma}
\newtheorem{proposition}[theorem]{Proposition}
\newtheorem{corollary}[theorem]{Corollary}

\newtheorem*{problem}{Problem}

\theoremstyle{remark}

\newtheorem{remark}[theorem]{Remark}

\theoremstyle{definition}
\newtheorem{definition}[theorem]{Definition}

\newcommand{\mc}[1]{\mathcal{#1}}
\newcommand{\spinorbundle}{\mathbf{S}}

\newcommand{\algo}{\textsc{PointForge} }

\author{Lisa Glaser \and Abel B. Stern}

\title{Reconstructing manifolds from truncated spectral triples}
\begin{document}

\maketitle
\begin{abstract}
  We explore the geometric implications of introducing a spectral cut-off on Riemannian manifolds.
  This is naturally phrased in the framework of non-commutative geometry, where we work with spectral triples that are \emph{truncated} by spectral projections of Dirac-type operators.
  We associate a metric space of `localized' states to each truncation.
  The Gromov-Hausdorff limit of these spaces is then shown to equal the underlying manifold one started with.
  This leads us to propose a computational algorithm that allows us to approximate these metric spaces from the finite-dimensional truncated spectral data.
  We subsequently develop a technique for embedding the resulting metric graphs in Euclidean space to asymptotically recover an isometric embedding of the limit.
  We test these algorithms on the truncated sphere and a recently investigated perturbation thereof.
\end{abstract}

\tableofcontents

\section{Introduction}
A natural notion of \emph{scale} is a major asset to any geometric theory with ties to the physical world.
After all, our geometric knowledge of objects appearing around us is finite, being limited by our observational power which fails us at high energy scales.
Additionally, the appearance of divergences in e.g. quantum field theory, especially when combined with gravity, is closely tied to bridging the gap between finite and infinite, or discrete and continuum, models.
Moreover, computational representation and analysis of geometric models strongly relies on our ability to extract from our model what is relevant and computationally feasible.

The field of noncommutative geometry has had close ties to physics ever since its inception, and yet lacks a consistent treatment of scale in the sense imagined.
The aim of this paper is to ameliorate this situation, by constructing a natural metric counterpart to the finite objects that are here referred to as \emph{truncated} (commutative) \emph{spectral triples}, and aiming to show that these do indeed carry enough information to describe their continuum limit in arbitrary detail.
For a wide-ranging and systematic approach to such truncations, see \cite{ConnesSuij}.

Admittedly, finite-dimensional objects in noncommutative geometry have enjoyed enduring attention.
General finite spectral triples have been classified~\cite{krajewski1998classification,cacic2011moduli, chamseddine2008standard} and parametrized~\cite{Barrett:2015naa}, and the Connes metric on these spaces has been studied in depth~\cite{iochum_distances_1999}.
However, this framework seems to lack simultaneous presence of 1) a natural link to the continuum in terms of metric spaces and 2) a natural link to the continuum in terms of spectral triples.

When representation-theoretic knowledge of the continuum is available, the framework of \emph{fuzzy spaces} such as those of \cite{Grosse1995, BALACHANDRAN2002184,Dolan_2003} seems to provide at least a natural link to the continuum in terms of spectral triples, and even some metric knowledge is available there~\cite{Schneiderbauer_Steinacker_2016}.
However, it might be said that the construction of a `fuzzy' version of a manifold is somewhat \emph{ad hoc} from a Riemannian viewpoint, and at least the framework has not yet seen successful extension to a reasonably large class of manifolds.

Truncated spectral triples (see Section~\ref{sec:trunc-spec-geom} and beyond) provide the advantage of a natural scale parameter and a natural symmetry-preserving correspondence between different scales.
The natural framework in which to study these truncations themselves as metric spaces is that of state spaces, as in \cite{ConnesSuij}, which can then be equipped either with the  Connes metric associated to the truncated spectral triple itself or with the pullback of the Connes metric on the full spectral triple.

An early and interesting study of the topological and metric properties of such spaces can be found in~\cite{d2014spectral}.
More recently, \cite{2020arXiv200508544V} investigates the question of Gromov-Hausdorff convergence of such truncated spaces under the truncated metric; see \cite{tey_torus_2019} for the example of the torus.
By contrast, here we are interested in \emph{localized} states, in order to recover e.g. a manifold $M$, rather than its metric space of probability measures.
Our companion paper~\cite{GlaserSternCCM19} investigates the relevance of the higher Heisenberg equation of~\cite{Chamseddine_Connes_Mukhanov_2015} in the framework of truncated spectral triples (see Section~\ref{sec:dc-intro}).

Arguably the main mathematical result of this paper is Corollary~\ref{cor:g-h-approximation}, which shows that the `localized' states of Section~\ref{sec:localized}, when equipped with the pullback metric, recover the Riemannian manifold one started with in the Gromov-Hausdorff limit.

In order to make the result more concrete and link back to the `computational representation' alluded to above, Section~\ref{sec:algorithm} is devoted to the description of an algorithm to approximate (finite subsets of) these metric spaces from the raw datum of a truncated spectral triple.
This allows us to test some of the results of Section~\ref{sec:localized} on the example of the sphere in Section~\ref{sec:example}.

Finite non-Euclidean metric spaces, as obtained by our algorithm, do not necessarily lend themselves to easy visualization or comparison by standard computational techniques.
In order to gain traction in this direction, we propose to look for new (asymptotically, locally isometric) embedding techniques and present a candidate approach in Section~\ref{sec:embedding_in_Rn}.
This allows us to visualize the metric results of Section~\ref{sec:example} and -- as originally inspired the technique -- compare the truncated spectral triple for $S^2$ and its perturbation as in~\cite{GlaserSternCCM19}.

\section{Truncated spectral geometries and point reconstruction} \label{sec:setup}
In noncommutative geometry one describes a spin manifold $M$ in terms of the associated spectral triple $(C^\infty(M), L^2(M, \spinorbundle), D_\spinorbundle)$, where $\spinorbundle$ is a spinor bundle over $M$ and $D_\spinorbundle$ is the corresponding Dirac operator.

Connes' reconstruction theorem \cite{Connes:SpectralCharacterizationManifolds} shows that this association is a bijection: one can fully reconstruct the underlying manifold $M$ and the chosen spin structure from the spectral triple alone.

Of particular interest for the present paper is the way one can recover the Kantorovich-Rubinstein distance between probability measures on $M$ from the interplay of the algebra $\mathcal{A} = C^\infty(M)$ and the Dirac operator $D = D_\spinorbundle$, acting on the Hilbert space of spinors.
Probability measures correspond to states on the \cstar-algebra $C(M) \supset C^\infty(M)$, which carries the topological, as opposed to differentiable, information about $M$.
Moreover, a function $f \in \mathcal{A}$ has Lipschitz constant $k$ if and only if $\norm{[D,f]} \leq k$ (as operators on $H = L^2(M, \spinorbundle_M)$).
Thus, Kantorovich-Rubinstein duality allows us to write
\begin{align}\label{eq:distance}
  d(\omega_1,\omega_2)=\sup_{a\in \mc{A}} \{ \left| \omega_1(a)-\omega_2(a) \right| \mid \norm{[D,a]} \leq 1\} \;.
\end{align}

When the states $\omega$ of $C(M)$ are \emph{pure}, they correspond to the atomic measures on single points, and we can thus recover $M$ with its metric.
This paper answers the question as to how we can understand this recovery of $M$ from the perspective of finite-dimensional parts of the representation of $\mathcal{A}$ and $D$ on $H$.
That is, we will construct natural counterparts to the ingredients above in the setting of \emph{truncated} spectral triples, and show that the metric space $M$ can be recovered as an asymptotic limit thereof.

\subsection{Truncated spectral geometries}
\label{sec:trunc-spec-geom}

From a mathematical perspective, it is desirable to be able to describe the infinite-dimensional datum of a spectral triple as a limit of finite-dimensional data of increasing precision, just like one can describe a compact Riemannian manifold as a Gromov-Hausdorff limit of finite metric spaces (by, for instance, equipping suitably dense finite subsets with the induced metric).
From a physical perspective, the same desire results from the view that one should be able to gain at least \emph{some} information about the geometry by probing it at finite energies.

A natural way to introduce such a `finite-energy cutoff', that is, truncation, of the geometric data $(A, H, D)$ is to pick a scale $\Lambda$, and define the corresponding spectral projection of $D$,
\[
  P_\Lambda \eqdef \chi_{[-\Lambda,\Lambda]}(D)
\]
to generate the finite-dimensional data
\begin{align}\label{eq:tSpT}
  (A_\Lambda,  H_\Lambda, D_\Lambda ) \eqdef (P_\Lambda A P_\Lambda, P_\Lambda H, P_\Lambda D )\;.
\end{align}
Truncated spectral triples were explored in depth by Connes and van~Suijlekom in ~\cite{ConnesSuij} and \cite{2020arXiv200508544V}; for a clear overview of the basic theory of their state spaces see~\cite{d2014spectral}.

This is the setting in which we wonder what counterpart to the duality $(A, H, D) \leftrightarrow M$, provided by the reconstruction theorem, can be found at the level of $(A_\Lambda, H_\Lambda, D_\Lambda)$.

\subsection{Point reconstruction}

We aim to refine the reconstruction of a spin manifold $M$ from its spectral triple $(A, H, D)$ through an understanding of the metric information contained in the truncations $(A_\Lambda, H_\Lambda, D_\Lambda)$.
The full manifold should then emerge asymptotically, such as through a Gromov-Hausdorff limit of the objects corresponding to the truncations.

The vector states of $C(M)$ that are induced by elements of $H_\Lambda$ appear naturally as vector states of $A_\Lambda$ as well, so that we have access to probability measures on $M$ directly in the truncated setting.
This, together with the distance formula \eqref{eq:distance}, is the main ingredient of our approach: we will identify states that correspond to points of $M$ in a suitable (asymptotic) sense, such that \eqref{eq:distance} asymptotically recovers the corresponding geodesic distance.

Our `proxy' approach to state localization, as discussed in section \ref{sec:local-y-disp}, was inspired by the notion of quasi-coherent states on fuzzy spaces defined in~\cite{Schneiderbauer_Steinacker_2016}.
However, as we aim just for the induced metric geometry on $M$ and view $(C^\infty(M)_\Lambda,H_\Lambda, D_\Lambda)$ rather as a finite observation of a spectral triple than as a quantization thereof, we will not construct coherent states in any quantum-mechanical sense but rather aim for localization only.
Moreover, as discussed below, this `proxy' approach is merely introduced to gain computational feasibility, at the expense of requiring identification of an embedding $\embedding: M \to \R^n$.
See also the end of Section~\ref{sec:metriconstates}.

After we define localized states, we prove existence of the desired objects: a sequence of metric spaces associated to the truncations $(A_\Lambda, H_\Lambda, D_\Lambda)$ that do indeed converge to $M$ in the Gromov-Hausdorff sense.
Then, Section~\ref{sec:algorithm} proposes an algorithm to construct these metric spaces  computationally, in order to make actual examples amenable to computer simulation.

In Section~\ref{sec:embedding_in_Rn} we propose a simple algorithm to obtain approximately locally isometric embedding of the resulting finite metric spaces into Euclidean space.
These should asymptotically converge to an isometric embedding of $M$ itself, and allow us to view the resulting finite metric spaces and investigate them more easily.

\section{The metric space of localized states}
\label{sec:localized}
Given a truncated commutative spectral triple $(C^\infty(M)_{\Lambda} ,H_{\Lambda} , D_{\Lambda} )$, we aim to construct a finite metric space that describes $M$ to the level of accuracy that the truncated spectrum will allow.

We will identify a \emph{subset} of the (vector) states of $C^\infty(M)_\Lambda$, consisting of those states that are \emph{localized} in a suitable sense and, therefore, correspond approximately to points of $M$.
The Connes metric on these vector states will then turn this subset into a metric space.

The guiding demand for this construction will be that the resulting
metric spaces should asymptotically (as $\Lambda \to \infty$ so that
$P_{\Lambda} \to 1$) converge, in the Gromov-Hausdorff sense, to the metric space $M$.

Now, the pure states of $C(M)$ -- that is, actual points in the metric space $M$ we are approximating -- do not necessarily extend to $C(M)_\Lambda$.
We do, however, have access to vector states induced by $v \in H_{\Lambda} $, which can be applied to either because both $C(M)$ and $C(M)_\Lambda$ are subsets of $B(H)$.

\begin{definition}
  $\statespace$ is the projective space over $H_\Lambda$.
\end{definition}

Each element $\brkt{v}$ of $\statespace$ corresponds to a positive linear functional of norm $1$ on $C(M)$ given by $a \mapsto \langle v, a v \rangle / \langle v, v \rangle$ (for any representative $v$ of $\brkt{v}$).
By the Riesz Representation theorem, such functionals correspond uniquely to probability measures on $M$.

\begin{definition}
  For $\brkt{v} \in \statespace$, $\mu_v$ is the unique probability measure such that $\langle v, a v \rangle / \langle v, v \rangle = \int_M a(x) d \mu_v(x)$ for all $a \in C(M)$ and any representative $v$ of $\brkt{v}$.
\end{definition}

By this identification, the Kantorovich-Rubinstein metric on the probability measures of $M$ induces a metric $\lambdist$ on $\statespace$.
It is an open conjecture that this metric can be (asymptotically) computed using the data $(C^\infty(M)_\Lambda, H_\Lambda, D_\Lambda)$: see Section~\ref{sec:metriconstates} for a discussion.

We will say that $\brkt{v}$ is localized when $\mu_v$ is sufficiently concentrated near a single point in $M$.
In order to quantify this notion, we will introduce the \emph{dispersion} functional $\disp$ on $\statespace$, below.

Let $\lambdist$ and $\contdist$ denote the metrics on $\statespace$ and $M$, respectively, so that we may regard both as subsets of the space of probability measures on $M$ equipped with the Kantorovich-Rubinstein metric.
We now wish to construct a subspace of $(\statespace, \lambdist)$ that is (Gromov-)Hausdoff close to $(M, \contdist)$.

Proposition~\ref{prop:statesarepoints} will show that there is a map $b \colon \statespace \to M$ such that $|\lambdist(\brkt{v}_1, \brkt{v}_2) - \contdist(b(\brkt{v}_1), b(\brkt{v}_2))| = O(\sqrt{\disp(\brkt{v}_1)} + \sqrt{\disp(\brkt{v}_2)})$; that is, our localized states can be identified almost isometrically with points of $M$.
Proposition~\ref{prop:pointsarestates} will then show that there is a corresponding asymptotically inverse map $\statemap \colon M \to \statespace$ such that $\contdist(x, b(\statemap(x))) = \tbigo{\Lambda^{-1}}$ and $\disp(\statemap(x)) = \tbigo{\Lambda^{-2} }$ uniformly in $x$.
This leads to Corollary~\ref{cor:g-h-approximation}, which shows that there exist a subspace $\estatespace{\epsilon^2}$ of $\statespace$ that is $\epsilon$-close to $M$ in Gromov-Hausdorff distance, where $\epsilon = \tbigo{\Lambda^{-1} }$.
Finally, Section~\ref{sec:metriconstates} discusses how these notions connect to the setting of $(C^\infty(M)_\Lambda, H_\Lambda, D_\Lambda)$.

\subsection{Localization: $\embedding$ and the dispersion functional}
\label{sec:local-y-disp}
Since elements of $\statespace$ correspond uniquely to probability measures on $M$, a natural way to measure the localization of such an element would be to take e.g. the variance of (isometrically embedded) position under this measure; that is, one would naturally define the dispersion of $\brkt{v} \in \statespace$ to be $\inf_{x \in M} E_{\mu_v} \left[ d(x, \cdot)^2 \right]$, where $E_{\mu_v}$ denotes expectation values under $\mu_v$.
In terms of the algebraic data, this quantity can be estimated as $\sup_{a \in C_\Lambda^\infty(M)}
\{ \langle v, a^* a v \rangle - |\langle v, a v \rangle|^2  \mid \norm{[D,a]} \leq 1 \}$.
However, the relevant non-convex double optimization problem -- to find minima $v$ of this dispersion in high-dimensional $H$ -- is computationally extremely challenging except in the simplest cases.

Therefore, we will require a proxy, $\embedding \colon M \to \R^n$, for the extremizing element $a$ above, in the sense that the Euclidean distance $\rndist{\embedding(x)}{\embedding(y)}$ on $M$ is bi-Lipschitz equivalent to the distance $\contdist(x, y)$ appearing in the variance.
Thus, let $\embedding: M \to \R^n$ be a (not necessarily Riemannian) embedding.

\begin{definition}
  Let $\mu$ be a probability measure on $M$.
  Then, the dispersion $\disp(\mu)$ equals
  \[
    \disp(\mu) \eqdef \int_M \rndist{\embedding(x)}{E_{\mu} \left[\embedding\right] }^2 d \mu(x)
  \]
\end{definition}

In probabilistic terms, $\disp(\mu)$ is just the trace of the covariance matrix of the vector-valued random variable $\embedding$, under the probability measure $\mu$.

\subsection{The $\phi$-barycenter of a localized state}
\label{sec:resolution}

An element $\brkt{v}$ of $\statespace$ that is considered to be localized should be localized \emph{somewhere}, that is, around some `barycenter' $x_v \in M$.
In order to control the localization of $\mu_v$ around the point $x_v$ by the dispersion $\disp(\mu_v)$, it is important that $\embedding(x_v)$ be close to $E_{\mu_v} \left[\embedding \right]$ in $\R^n$.
Hence,

\begin{definition}
  Let $\mu$ be a probability measure on $M$.
  Then a \emph{$\phi$-barycenter} of $\mu$ is any point $x \in M$ that minimizes $\rndist{\embedding(x)}{E_{\mu} \left[ \embedding \right]}$.
\end{definition}

By compactness of $M$ and continuity of $\rndist{\embedding(\cdot)}{E_{\mu} \left[\embedding \right]}$, there always exists a $\phi$-barycenter.

Localized states are indeed concentrated near their $\phi$-barycenters, as the following lemma shows.
That is, the dispersion $\disp(\mu)$ is a good proxy for the squared second Wasserstein distance\footnote{See e.g. \cite{villani2003topics} for an introduction to the measure-theoretic notions that are applied (without any hint of sophistication) in this section.} $W_2(\mu, \delta_{x})^2$ between the measure $\mu$ and any given $\phi$-barycenter $x$ thereof.

\begin{lemma}
  \label{lem:localized-near-barycenter}
  Any $\phi$-barycenter $x$ of a probability measure $\mu$ satisfies
  \[
    W_2(\mu, \delta_x)^2 \eqdef \int_M \contdist(z, x)^2 d\mu(z) = O(\disp(\mu)),
  \]
  uniformly\footnote{That is, the relevant constant depends only on $\embedding$ and $M$, not on $\mu$.} in $\mu$, where $\delta_{x}$ denotes the Dirac measure centered on $x$.
  Moreover, any two $\phi$-barycenters of $\mu$ are within distance $O(\sqrt{\disp(\mu)})$, uniformly in $\mu$, of each other.
\end{lemma}
\begin{proof}
  By Chebyshev's inequality, $\mu (\left\{ x \in M \mid \norm{\embedding(x) - E_{\mu}\left[\embedding \right]} \geq t \right\})$ is bounded by $t^{-2} E_{\mu} \left[ \norm{\embedding - E_{\mu} \left[\embedding \right]}^2 \right] = t^{-2} \disp(\mu)$.
  Therefore, if $\rndist{\embedding(\cdot)}{E_{\mu} \left[\embedding \right]}^2 \geq t$ on the support of $\mu$, we see that $1 \leq t^{-2} \disp(\mu)$.
  Therefore, we conclude that $\inf_{x \in \operatorname{supp} \mu} \rndist{\embedding(x)}{E_{\mu} \left[ \embedding \right]} \leq \sqrt{\disp(\mu)}$.
  Any $\phi$-barycenter of $\mu$ must therefore, as a minimizer of $\rndist{\embedding(\cdot)}{E_{\mu} \left[\embedding \right]}$ in $M \supset \operatorname{supp} \mu$, also satify this inequality.
  
  Now, as a smooth embedding, $\embedding$ is automatically bi-lipschitz.
  In particular, there exists $\beta$ such that $\contdist(x, y) \leq \beta \rndist{\embedding(x)}{\embedding(y)}$ uniformly in $x,y$.
  We see, therefore, that any two $\phi$-barycenters of $\mu$ are at a distance at most $2 \beta \sqrt{\disp(\mu)}$.
  
  Moreover, for all $x \in M$, we conclude that $\int_M \contdist(z, x)^2 d \mu(z) $ is bounded by $\beta^2 \int_M \rndist{\embedding(z)}{\embedding(x)}^2 d \mu(z)$.
  As $| \int_M f^2 - g^2 d \mu | \leq \int_M (2|g| + |f-g|)|f-g| d \mu$ and for $\phi$-barycenters $x$ of $\mu$, $\left| \rndist{\embedding(z)}{\embedding{x}} - \rndist{\embedding{z}}{E_{\mu}\left[\embedding\right]} \right| \leq \sqrt{\disp(\mu)}$, the error $| \eta(\mu) - \rndist{\embedding(z)}{\embedding(x)}^2 d \mu(z) |$ is bounded by $\int_M (2 \rndist{\embedding(z)}{E_{\mu}\left[\embedding\right]} + \sqrt{\disp(\mu)}) \sqrt{\disp(\mu)} d \mu(z)$ which is, by the classical Jensen inequality and the definition of $\eta(\mu)$, bounded by $3 \disp(\mu)$.
  The Lemma follows.
\end{proof}

We now consider the implications of the above for the barycenters of probability measures $\mu_v$, for $\brkt{v} \in \statespace$.

\begin{proposition} \label{prop:statesarepoints}
  There exists a map $b \colon \statespace \to M$ such that
  \[
    \left| \lambdist(\brkt{v}, \brkt{w}) - \contdist(b(\brkt{v}), b(\brkt{w})) \right| = O(\sqrt{\disp(\mu_v)} + \sqrt{\disp(\mu_w)})
    \]
  as $\disp(\mu_v), \disp(\mu_w) \to 0$, uniformly in $\brkt{v}$, $\brkt{w}$.
\end{proposition}
\begin{proof}
  Let $b$ assign a choice of $\phi$-barycenter to each $\mu_v$, $\brkt{v} \in \statespace$.

  Now let $\delta_{v}, \delta_w$ be the Dirac measures centered on $b(\brkt{v})$, $b(\brkt{w})$.
  Recall that the distance $\lambdist(\brkt{v}, \brkt{w})$ is the Kantorovich-Rubinstein distance $W_1(\mu_v, \mu_w)$ between $\mu_v$ and $\mu_w$, and $\contdist(b(\brkt{v}), b(\brkt{w})) = W_1(\delta_v, \delta_w)$.

  By the triangle inequality for the metric $W_1$, $W_1(\mu_v, \mu_w) - W_1(\delta_v, \delta_w) \leq W_1(\mu_v, \delta_v) + W_1(\delta_w, \mu_w)$ and similarly $W_1(\delta_v, \delta_w) - W_1(\mu_v, \mu_w) \leq W_1(\delta_v, \mu_v) + W_1(\mu_w, \delta_w)$, so that $\left| W_1(\mu_v, \mu_w) - W_1(\delta_v, \delta_w) \right| \leq W_1(\mu_v, \delta_{v}) + W_1(\mu_w, \delta_{w})$.
  
  Now, by the classical Jensen inequality $\int_M |f| d \mu \leq \sqrt{\int_M |f|^2 d \mu}$ we have $W_1(\mu_v, \delta_{v}) \leq W_2(\mu_v, \delta_{v})$ and we conclude that
  \begin{align*}
    \left| \lambdist(\brkt{v}, \brkt{w}) - \contdist(b(\brkt{v}), b(\brkt{w})) \right| &= \left| W_1(\mu_v, \mu_w) - W_1(\delta_v, \delta_w) \right| \\
    &\leq W_1(\mu_v, \delta_{v}) + W_1(\mu_w, \delta_{w}) \\
                                                                  &\leq W_2(\mu_v, \delta_{v}) + W_2(\mu_w, \delta_{w}) \\
    &= O(\sqrt{\disp(\mu_v)} + \sqrt{\disp(\mu_w)}),
  \end{align*}
  where the last line is Lemma~\ref{lem:localized-near-barycenter}.
\end{proof}

\subsection{Existence of localized states near any point}
\label{sec:explicit-small-dispersion}
Proposition~\ref{prop:statesarepoints} tells us that probability measures $\mu$ on $M$ of sufficiently small dispersion correspond well to their $\phi$-barycenters $b(\mu)$.
This holds in particular for the probability measures $\mu_v$ associated to $\brkt{v} \in \statespace$.
We would now like to estimate the converse, i.e. to show that each point $x$ corresponds to an element $\statemap(x) \in \statespace$ whose probability measure is of small dispersion.
Then, with $\statemap$ asymptotically an isometric embedding and $b \circ \statemap$ asymptotically the identity on $M$, we would rightly be able to say there is a picture of $M$ inside $\statespace$.

To simplify the asymptotic estimates we will introduce the notation $\tbigo{}$ common in computer science:

\begin{definition}
  Let $X$ be a set and consider functions $f \colon X \times \R_+ \to \C$, $g \colon X \times \R_+ \to \R_+$.
  We say that $f = \bigo{g}$ uniformly when there exist finite $C, r_0 > 0$ such that  $|f(x, r)| \leq C g(x, r)$ for all $r > r_0$ and all $x \in X$.
  We say that $f = \tbigo{g}$ uniformly when $f = \bigo{g |\log g|^s}$ uniformly for some $s \geq 0$.
\end{definition}

\begin{proposition} \label{prop:pointsarestates}
  Let $M$ be a \spinc{} manifold equipped with a Dirac-type operator $D$ on a Hermitian vector bundle $\pi \colon \spinorbundle \to M$.
  Then, there exists a family $\{\statemap\}_\Lambda$ of maps $\statemap: \projective{\spinorbundle} \to \statespace$ such that for all $\epsilon > 0$,
  \begin{itemize}
  \item $\lambdist(\statemap(v), \statemap(w)) = \contdist(\pi(v), \pi(w)) + \tbigo{\Lambda^{-1 }}$ uniformly.
  \item The dispersion $\disp(\mu)$ of the measure $\mu$ associated to $\statemap(v)$ is $\tbigo{\Lambda^{-2 }}$ uniformly.
  \item The maps $\statemap$ asymptotically invert $b$, in the sense that $\contdist(\pi(v), b(\statemap(v))) = \tbigo{\Lambda^{-1 }}$ uniformly and $\lambdist(\statemap(v)), \brkt{v}) = \tbigo{\sqrt{\disp(\mu_v)} + \Lambda^{-2}}$ uniformly whenever $b(\brkt{v}) = \pi(\statemap(v))$.
  \end{itemize}
\end{proposition}

The proof depends on a balanced rescaling of the truncated heat flow and will be presented at the end of this section.

To discuss the geometric consequences of Proposition~\ref{prop:pointsarestates}, we will introduce notation for the small-dispersion subset into which $\statemap$ maps.

\begin{definition}
  Let $\epsilon > 0$.
  Then, $\estatespace{\epsilon} \subset \statespace$ consists of those $\brkt{v} \in \statespace$ for which $\disp(\mu_v) < \epsilon$.
\end{definition}

\begin{corollary} \label{cor:g-h-approximation}
  As $\Lambda \to \infty$, there exists $\epsilon = \tbigo{\Lambda^{-1}}$ such that the Gromov-Hausdorff distance between $M$ and the space $\estatespace{\epsilon^2}$, equipped with the metric $\lambdist$, is $O(\epsilon)$.
\end{corollary}
\begin{proof}
  Let $\epsilon^2 = \sup_{x \in M} \disp(\mu_{\phi_\Lambda(x)})$.
  By the second part of Proposition~\ref{prop:pointsarestates}, $\epsilon^2 = \tbigo{\Lambda^{-2} }$.

  Now, the map $\brkt{v} \to \mu_v$ sends $\estatespace{\epsilon^2}$ isometrically into the space of probability measures on $M$ with the Kantorovich-Rubinstein metric $W_1$, and the map $x \mapsto \delta_x$ sends $M$ isometrically into the same space.

  For $\brkt{v} \in \estatespace{\epsilon^2}$, there is a point $x = b(\brkt{v})$ in $M$ such that $W_1(\mu_v, \delta_x) = O(\epsilon)$ by Proposition~\ref{prop:statesarepoints}.
  
  For $x \in M$, there is an element $\brkt{v} = \statemap(x)$ of $\estatespace{\epsilon^2}$ that satisfies $\contdist(x, b(\brkt{v})) = \tbigo{\Lambda^{-1}} = O(\epsilon)$ by the third part of Proposition~\ref{prop:pointsarestates}.
  Let $\delta_v$ be the Dirac measure centered at $b(\brkt{v})$, so that $W_1(\delta_v, \delta_x) = \contdist(x, b(\brkt{v}))$.
  Now, $W_1(\mu_v, \delta_x) \leq W_1(\mu_v, \delta_v) + W_1(\delta_v, \delta_x)$, and $W_1(\mu_v, \delta_v) = O(\sqrt{\disp(\mu_v)}) = O(\epsilon)$ by Proposition~\ref{prop:statesarepoints}, so that indeed $W_1(\mu_v, \delta_x) = O(\epsilon)$.
  We conclude that the Hausdorff distance between $M$ and $\estatespace{\epsilon^2}$, as subsets of the space of probability measures on $M$, is $O(\epsilon)$.
\end{proof}

\subsubsection{Proof of Proposition~\ref{prop:pointsarestates}}
Recall the spectral triple $(C^\infty(M), L^2(M, \spinorbundle), D)$ associated to $M$, where $\spinorbundle$ is a spinor bundle over $M$ and $D$ is a Dirac-type operator on $\spinorbundle$.

We will define a localization map $F_\Lambda \colon \spinorbundle \to P_\Lambda H$ such that

\[
  \ip{F_\Lambda(v_x), a F_\Lambda(w_x)} = \spip{v_x, a_x w_x} + \norm{v_x} \norm{w_x} O \left(\norm{a} \Lambda^{-2 } + \lip{k}{a}{x} \Lambda^{-k } \right)
\]
for all $\epsilon > 0$, whenever $v_x, w_x \in \spinorbundle_x$ and $a \in \Gamma(\End \spinorbundle)$.
If $\Psi_{xy}$ is the parallel transport map from $\spinorbundle_y$ to $\spinorbundle_x$, the constant $\lip{k}{a}{x}$ is defined to be $\sup_{y \colon d(x,y) \leq \rho} \norm{a_x - \Psi_{xy}^* a_y} / d(x, y)^k$.

To do so, we will take the element $v_x \in \spinorbundle_x$ and use the short-time heat flow associated to the Laplace-type operator $D^2$ to obtain a smooth section $y \mapsto \pk{t}{x}{y}(v_x)$ of $\spinorbundle$, which then corresponds to an element of $H$.
The known estimates on heat asymptotics will allow us to bound the dispersion of $\pk{t}{x}{y}(v_x)$ for small $t$.
Then, the fact that $\pliftmap{t}$ is the heat kernel associated to $D^2$ whereas $P_\Lambda$ is an associated projection, will allow us to control the behaviour of $(1 - P_\Lambda) \pk{t}{x}{y}(v_x)$.

\begin{definition}
  Let $v_x \in \spinorbundle_x$, for $x \in M$.
  Then $\plift{t}{v_x}$ is the section $y \mapsto \pk{t}{x}{y}(v_x)$ of $\spinorbundle$, where $\pliftmap{t}$ is the integral kernel associated to the operator $e^{- t D^2}$.
\end{definition}

The following Lemma allows us to control the leading term in the short-time behaviour of the heat flow $\plift{t}{v_x}$.
To that end, let $\hk{t}{x}{y}$ equal the scalar coefficient $e^{-d_M(x, y)^2 / 4t} (4 \pi t)^{-m/2}$ of the leading term in the asymptotics of the heat kernel.
For $x \in M$ and $s \in \R$, let $B_s(x) \subset M$ be the metric ball of radius $s$ around $x$.

\begin{lemma}
  \label{lemma:heat-kernel-fiberwise-inner-product}
  Let $a \in \Gamma(\End \spinorbundle)$.
  Then, we have for all $s$ smaller than the injectivity radius of $M$, and all $v, w \in \Gamma(\spinorbundle)$,
  \begin{align*}
    \int_{B_s(x)} \hk{t}{x}{y} \spip{v_x,  \Psi_{xy}^* a_y w_x} d y =& \spip{v_x, a_x w_x} \int_{B_s(x)} \hk{t}{x}{y} dy + \\
    &+ \norm{v} \norm{w} \lip{k}{a}{x} O(t^{k/2} + s^{-2} t^{(k+2)/2}), \\
    \int_{B_s(x)} \hk{t}{x}{y} dy =& 1 + O(t + s^{-4} t^{2}),
  \end{align*}
  uniformly in $v_x, a, x \in M$.
  
\end{lemma}
\begin{proof}
  For $k \geq 0$, consider the integral $m_{t, s, k}(x) \eqdef \int_{B_s(x)} \hk{t}{x}{y} d^{k}(x, y) dy$.
  Let $m'_{t, s, k} \eqdef (4 \pi t)^{-m/2} \int_{\norm{y} \leq s} e^{- \norm{y}^2/4t} \norm{y}^{k} dy$.
  There exists a global constant $C$ such that the pullback of the volume form on $M$ is bounded by $C \norm{y}^2$ times the Euclidean volume form.
Thus, pulling back our integral through the exponential map at $x$, we have $|m_{t, s, k}(x) - m'_{t, s, k}| \leq C m'_{t, s, k+2}$.

By Chebyshev's inequality, we have $m'_{t, s, 2k} = (4 \pi t)^{-m/2} \int e^{- \norm{y}^2 / 4t} \norm{y}^{2k} dy + O(t^{-m/2} s^{-4} \int_{\complement B_s(0)} e^{- \norm{y}^2/4t} \norm{y}^{2k+4} dy)$, where $\complement$ denotes the complement.
With Isserlis' theorem to calculate the full Gaussian integrals, we see that $m'_{t, s, 2k} =  c_k t^{k} + O(s^{-4} t^{k+2})$ for all $k$. Thus,$m_{t, s, 2k}(x) = c_k t^{k} + O(t^{k+1}) + O(s^{-4} t^{k+2})$.

Now, estimate $\left| \spip{v_x, \Psi_{xy}^* a_y w_x} - \spip{v_x, a_x v_x} \right| \leq
\lip{k}{a}{x} d(x, y)^k \norm{v_x} \norm{w_x}$.

As $m_{t, s, k}^2 \leq m_{t, s, 0} m_{t, s, 2k}$ by the classical Jensen's inequality, we conclude that, with $\sqrt{m_{t, s, 0}(x) m_{t, s, k}(x)} = O(t^{k/2} + s^{-2} t^{k/2})$, the remaining error is $O(\norm{v_x} \norm{w_x} \lip{k}{a}{x} (t^{k/2} + s^{-2} t^{(k+2)/2}))$ uniformly.
\end{proof}

We are now in a position to show that the rescaled heat flow $(2 \pi t)^{m/4} \pliftmap{t} \colon \spinorbundle_x \to H$ is asymptotically isometric, in the following sense:

\begin{lemma}
  \label{lemma:global-heat-kernel-lift-inner-product}
  For $a \in \Gamma(\End \spinorbundle)$ and $v,w \in \Gamma(s)$, we have uniformly
  \begin{align*}
    \ip{\plift{t}{v_x}, a \plift{t}{w_x}} =& (2 \pi t)^{- m / 2} \spip{v_x, a_x w_x}  + \\
                               &+ \norm{v_x} \norm{w_x} O(\lip{k}{a}{x} t^{(k-m)/2} + \norm{a} t^{(2-m)/2})
\end{align*}
\end{lemma}
\begin{proof}

  It is well-known, see e.g. \cite[Theorem 2.30]{MR2273508}, that there exist a nonzero radius $s$ around $x$ such that for $d_{M}(x, y) < s$, one has $\pk{t}{x}{y}(v_x) = \hk{t}{x}{y}(\Psi_{xy}(v_x) + O(t))$   as $t \to 0$, where $\Psi$ is the parallel transport along the spinor connection.
  Therefore,
  \[
    \spip{\pk{t}{x}{y}(v_x), a_y \pk{t}{x}{y}(w_x)} = (\hk{t}{x}{y})^2 \spip{v_x, \Psi_{xy}^* a_y w_x} + O(t (\hk{t}{x}{y})^2 \norm{a} \norm{v_x} \norm{w_x}),
  \]
  uniformly in $x$.
  Moreover, for $d_{M}(x, y) > s$, one has $\spip{\pk{t}{x}{y}(v_x), a_y \pk{t}{x}{y}(w_x)} = O((\hk{t}{x}{y})^2 \norm{a} \norm{v_x} \norm{w_x})$.

  Now, outside an $s$-ball around $x$, we have
  \[
    \int_{\complement B_s(x)}     \spip{\pk{t}{x}{y}(v_x), a_y \pk{t}{x}{y}(w_x)} dy  = O(e^{s^2/2t} t^{-m} \norm{a} \norm{v_x} \norm{w_x}),
  \]
  as $t \to 0$.  Now set $s \eqdef t^{1/4}$ and note that $(\hk{t}{x}{y})^2 = (2 \pi t)^{-m/2} \hk{2t}{x}{y}$.
  The estimate of Lemma~\ref{lemma:heat-kernel-fiberwise-inner-product} on the integral over $B_s(x)$ then completes the proof.
\end{proof}

In order to estimate the scaling of the truncated heat flow $P_\Lambda \plift{t}{v_x}$ with $\Lambda$, we will relate $\plift{t}{v_x}$ to the spectral resolution of $D^2$, as follows.

\begin{definition}
  Let $P_\lambda$ be the projection onto the $\lambda$-eigenspace of the first-order elliptic differential operator $D$ and let $\spectralkernel{\lambda}$ be its integral kernel, so that for all sections $v$ of $\spinorbundle$ we have
  \[
    P_\lambda(v)(y) = \int_M \spectralkernel{\lambda}(x, y) (v_x) dx.
  \]
  Then, $\spectrallift{\lambda} \colon \spinorbundle \to \Gamma(\spinorbundle)$ is the associated lifting $\spectrallift{\lambda}(v_x) \colon y \mapsto \spectralkernel{\lambda}(x, y)(v_x)$.
\end{definition}

In particular, we have $\pliftmap{t} = \sum_\lambda e^{-t \lambda^2} \spectrallift{\lambda}$ weakly.
To estimate the $L^2$-norm of $\spectrallift{\lambda}(v_x)$, we will need the following classical result by Hörmander\cite{MR609014}.

\begin{theorem}[{\cite[Theorem 4.4]{MR609014}}]
  \label{theorem:sup-bound-on-evals}
  There exists a constant $C$ such that
  \[
    \sup_{x, y \in M} \norm{\spectralkernel{\lambda}(x, y)} \leq C (1 + |\lambda|)^{\dim M - 1}
  \]
  uniformly in $x, y, \lambda$. In particular, there exists a constant $c$ such that
  \[
    \norm{\spectrallift{\lambda}(v_x)}_H \leq c (1 + |\lambda|)^{\dim M - 1} \norm{v_x}_{\spinorbundle_x},
  \]
  for all $v_x \in \spinorbundle$, all $x \in M$ and all $\lambda \in \sigma(D) \subset R$.
\end{theorem}

Due to the polynomial scaling of the fiberwise inner product, we can now show that the exponential dependence of $(1 - P_\Lambda) \plift{t}{v_x}$ on $\Lambda$ implies that we retain the asymptotic properties of the heat flow when we truncate such that $\Lambda^{2} t = c \log \Lambda$ for sufficiently large $c$.

\begin{lemma}
  \label{lemma:heat-kernel-versus-truncation-global-inner-product}
  We have for all $a \in B(H)$,
  \[
    \left| \ip{\plift{t}{v_x}, (a - P_\Lambda a P_\Lambda) \plift{t}{w_x}}_H \right| = O(\norm{v_x} \norm{w_x} \norm{a} t^{1-2m} e^{-t z \Lambda^2}),
  \]
  for all fixed $0 \leq z < 1$, uniformly in $v_x, w_x \in \spinorbundle_x, x \in M$, as $\Lambda \to \infty$.
\end{lemma}
\begin{proof}
    Recall that the integral transform associated to the kernel $\pk{t}{x}{y}$ equals the bounded linear operator $w \mapsto e^{-t D^2} w$ on $H$, so that $\pliftmap{t} = \sum_\lambda e^{-t \lambda^2} \spectrallift{\lambda}$ weakly.
    Thus,  $\langle P_\Lambda \plift{t}{v_x}, w \rangle = \sum_{|\lambda| < \Lambda \in \sigma(D)} e^{-t \lambda^2} \ip{\spectrallift{\lambda}(v_x), w }$.

  The difference to be estimated, then, consists of the sum of terms missing in $\left\langle P_\Lambda \plift{t}{v_x}, a P_\Lambda \plift{t}{w_x} \right\rangle$, which equals $\sum_{\lambda_1, \lambda_2 \not \in [-\Lambda,\Lambda]^2} e^{-t(\lambda_1^2 + \lambda_2^2)} \ip{\spectrallift{\lambda_1}(v_x), a \spectrallift{\lambda_2}(w_x)}$.

  First note that Theorem~\ref{theorem:sup-bound-on-evals} provides a global constant $c$ that bounds the factor $|\ip{\spectrallift{\lambda_1}(v_x), a \spectrallift{\lambda_2}(w_x)}|$ by $\norm{v_x} \norm{w_x} \norm{a} c^2 \left((1 + \lambda_1^2)(1 + \lambda_2^2)\right)^{m -1}$.

  Now, $\sum_{|\lambda| > \Lambda} e^{-t \lambda^2} (1 + \lambda^2)^{m-1}$ is, for $0 < \epsilon \leq 1$, bounded by $e^{-(1-\epsilon)t \Lambda^2}$ times the shifted sum $\sum_{|\lambda| > \Lambda} e^{-t \epsilon \lambda^2} (1 + \lambda^2)^{m-1}$.
  Moreover, the entire sum $\sum_{\lambda} e^{-t \lambda^2} ( 1 + \lambda^2)^{m-1}$ is, by the heat asymptotics for the Laplace-type operator $D^2$, bounded by $O(t^{\frac12 -m})$.

  Thus, we obtain a bound of $\bigo{c^2 \norm{v_x} \norm{w_x} \norm{a} t^{1-2m} e^{-t(1-\epsilon) \Lambda^2}}$.
\end{proof}

Our localization map is thus given by a truncated, rescaled heat flow, as follows.

\begin{definition}
  Let $t_\Lambda \eqdef  2m \Lambda^{-2} \log \Lambda$.
  The map $F_\Lambda \colon \spinorbundle \to P_\Lambda H$ is given by
  \[
    v_x \mapsto (2 \pi t_\Lambda)^{m/4} \sum_{|\lambda| \leq \Lambda} e^{-t_\Lambda \lambda^2} \spectrallift{\lambda}(v_x).
  \]
\end{definition}

There exists finite $\Lambda$ such that $F_\Lambda$ is injective, by Lemma~\ref{lemma:global-heat-kernel-lift-inner-product} and injectivity of the heat flow $\plift{t}{v_x}$.

Now we are in a position to connect $F_\Lambda$ to the localization question of Proposition~\ref{prop:pointsarestates}.

\begin{proposition} \label{proposition:spinor-states-localized-near-base-point}
  Consider the map $\statemap \colon \mathbb{P} \spinorbundle \to \statespace$ given by 
  \[
    \statemap([v_x]) \eqdef \left[ F_\Lambda(v_x) \right] \in \statespace,
  \]
  for $\Lambda$ sufficiently large that $F_\Lambda$ is injective.
  Then, $\statemap([v_x])$ is localized near $x$ in the sense that
  \begin{align*}
    \disp(\mu_{\statemap([v_x])}) = \tbigo{\Lambda^{-2 }}, && W_2(\mu_{\statemap([v_x])}, \delta_x)^2 = \tbigo{\Lambda^{-2 }}.
  \end{align*}
  \end{proposition}
  \begin{proof}
    Note that for any $\epsilon > 0$ we may pick $z$ such that $t_\Lambda^{1 - 2m} e^{- t_\Lambda z \Lambda^2} = \tbigo{\Lambda^{-2 }}$.
    Thus, for $a \in \Gamma(\End \spinorbundle)$ we have
        \begin{align*}
          \ip{F_\Lambda(v_x), a F_\Lambda(v_x)} &=
                                                  (2 \pi t)^{m/2} \ip{\plift{t}{v_x}, a \plift{t}{v_x}} + \norm{v_x}^2 \norm{a} \tbigo{\Lambda^{-2 }} \\
          &= (v_x, a_x v_x) + \norm{v_x}^2 \tbigo{\lip{k}{a}{x} \Lambda^{-k } + \norm{a} \Lambda^{-2 }}
          \end{align*}
          so that in particular $\norm{F_\Lambda(v_x)}^2 = \norm{v_x}^2 (1 + \tbigo{\Lambda^{-2 }})$.
          
          With $f_x(y) \eqdef d(x, y)^2$ and $\norm{v_x}^2 = 1$, we have $W_2(\mu_{\statemap([v_x])}, \delta_x)^2 = \statemap([v_x])(f_x) = \tbigo{\Lambda^{-2 }}$, and with $g(y)_i \eqdef \embedding(y)_i - \embedding(x)_i$ we see that $\statemap([v_x])(g_i)$ is $\tbigo{\Lambda^{-1 }}$. The dispersion of the associated measure is therefore $\tbigo{\Lambda^{-2 }}$.
\end{proof}

\begin{proof}[Proof of Proposition~\ref{prop:pointsarestates}]
  Let $\mu_1, \mu_2$ be the measures associated to $\statemap(v_x)$ and $\statemap(w_y)$ respectively.
  Then, $\disp(\mu_i) = \tbigo{\Lambda^{-2 }}$ by Proposition~\ref{proposition:spinor-states-localized-near-base-point} and we have $\left| \contdist(x, y) - d(\mu_1, \mu_2) \right| \leq W_2(\delta_x, \mu_1) + W_2(\mu_2, \delta_y) = \tbigo{\Lambda^{-1 }}$ .

  Let $p \eqdef E_{\mu_{\statemap(v_x)}} \left[ \embedding \right]$.
  Then, $\rndist{\embedding(x)}{p} \leq W_2(\embedding(x), \embedding_* \mu_1) + \sqrt{\disp(\mu)}$ and the first term is, by bi-Lipschitz equivalence, $O(W_2(x, \mu_1))$ so that $\rndist{\embedding(x)}{p} = \tbigo{\Lambda^{-1 }}$.
  Thus, $d(x, b(\statemap(v_x))) = \tbigo{\Lambda^{-1 }}$ as well.

  Finally, for probability measures $\nu$, $x = b(\nu)$ and $0 \neq v \in \spinorbundle_x$, we have 
$d(\statemap(v), \nu) \leq W_2(\delta_x, \nu) + W_2(\statemap(v), \delta_x)$, and Lemma~\ref{lem:localized-near-barycenter} leads to bounds of $O(\sqrt{\disp(\nu)})$ and $\tbigo{\Lambda^{-1 }}$ (when combined with Proposition~\ref{proposition:spinor-states-localized-near-base-point}), respectively, on the latter.
\end{proof}

\subsection{The space $\statespace$ in terms of the truncated spectral triple}
\label{sec:metriconstates}
By Corollary~\ref{cor:g-h-approximation}, there exists $\epsilon = \tbigo{\Lambda^{-1} }$ such that the space $\estatespace{\epsilon^2}$ is $\epsilon$-close to $M$, when equipped with the Kantorovich-Rubinstein metric.
By Kantorovich-Rubinstein duality, that metric can be computed by Connes' formula \eqref{eq:distance}:
\begin{align*}
  \lambdist(\brkt{v}, \brkt{w})
  &=
    \sup_{f \in C^\infty(M)} \left\{ \left| \int_M f d \mu_v - \int_M f d \mu_w \right| \mid \norm{[D,f]} \leq 1\right\}
\end{align*}

Now, each element $\brkt{v}$ of $\statespace$ induces a state $\omega_v$ of $C^\infty(M)_\Lambda$, which corresponds to representatives $v \in H_\Lambda$ of $\brkt{v}$ as
\begin{align*}
  \omega_v(P_\Lambda f P_\Lambda) &= \langle v, P_\Lambda f P_\Lambda v \rangle / \norm{v}_H^2 = \langle v, f v \rangle / \norm{v}_{H}^2
  \\
  &= \int_M f(x) d \mu_v(x).
\end{align*}

With this identification, we have
\begin{align} \label{eq:pullback-distance}
  \lambdist(\brkt{v}, \brkt{w})
  &=
    \sup_{f \in C^\infty(M)} \left\{ \left| \omega_v(P_\Lambda f P_\Lambda) - \omega_w(P_\Lambda f P_\Lambda) \right| \mid \norm{[D,f]} \leq 1 \right\}.
\end{align}

It is an open question whether, in the limit $\Lambda \to \infty$, this metric can be approximated arbitrarily well in Gromov-Hausdorff distance by the functional on the \emph{truncated} spectral triple given by

\begin{align} \label{eq:truncated-distance}
  \trlambdist(\brkt{v}, \brkt{w})
  &=
    \sup_{f \in C_\Lambda^\infty(M)} \left\{ \left| \omega_v(f) - \omega_w(f) \right| \mid \norm{[D_\Lambda, f ]} \leq 1 \right\}.
\end{align}

Although we clearly have $\norm{[D_\Lambda, P_\Lambda f P_\Lambda]} \leq \norm{[D,f]}$ so that $\lambdist \leq \trlambdist$, it is a highly nontrivial undertaking to obtain a bound in the \emph{opposite} direction.
See e.g. \cite{d2014spectral} for further perspective on the problem.

The first definitive result in this direction is that of \cite{2020arXiv200508544V},  where it is shown that the appropriately chosen state spaces converge in Gromov-Hausdorff distance whenever there exists a \emph{$C^1$-contractive approximate order isomorphism}, that is, a map $S \colon C_\Lambda^\infty(M) \to C^\infty(M)$ such that $\norm{S(P_\Lambda f P_\Lambda) - f}$ and $\norm{P_\Lambda(S(f_\Lambda))P_\Lambda - f_\Lambda}$ are bounded by $o(1)$ multiples of $\norm{[D,f]}$ and $\norm{[D_\Lambda, f_\Lambda]}$ respectively, and that moreover satisfies $\norm{[D,S(f_\lambda)]} \leq \norm{[D_\Lambda, f_\Lambda]}$, $\norm{S(f_\Lambda)} \leq \norm{f_\Lambda}$.

On the experimental side, the comparison between these distances for certain localized states on the sphere in Section~\ref{sec:sphere_graph_error_analysis} presents corroborating evidence to the idea that the Gromow-Hausdorff limit of $\statespace$ should not depend on whether we equip it with $\lambdist$ or with $\trlambdist$.

\begin{remark}
  In the present context we \emph{do} have a natural map from $B(P_\Lambda H)$ to $\Gamma^\infty(\End \spinorbundle)$ defined by the bilinear form $(v_x, S(a_\Lambda) w_x) \eqdef \ip{F_\Lambda(v_x), a_\Lambda F_\Lambda(w_x)}$ for $a_\Lambda \in B(P_\Lambda H)$.
  If we then take the normalized fiberwise trace of the corresponding endomorphism, we obtain a smooth function on $M$.
  The norm of this map is asymptotically equivalent to $1$, but the estimates of the previous section just show that\footnote{Note that the $\norm{a}$ coefficient in the error bound is due to the nonleading terms in the heat expansion, so that it vanishes for flat manifolds.}, for $a \in \Gamma(\End \spinorbundle)$, $\norm{a - S(P_\lambda a P_\Lambda)} = \tbigo{\norm{a} \Lambda^{-2 } + \sup_x \lip{1}{a}{x} \Lambda^{-1 }}$ (so that, \emph{a fortiori}, we have a similar bound on $\norm{a_\Lambda - S(a_\Lambda)}$), and more importantly it is not clear that we can control $\norm{[D, S(f_\Lambda)]}$ by $\norm{[D_\Lambda, f_\Lambda]}$ due to the presence of a boundary term.
\end{remark}

It is still very well possible that the theorem holds in general, but attacking this difficult problem is beyond the scope of the present paper.
We do conjecture that it is possible to define a metric on the spaces $\estatespace{\epsilon_\Lambda^2}$ that converges to $M$ in Gromov-Hausdorff limit as $\Lambda \to \infty$, perhaps by using a dual scale where $(C_{\Lambda'}^\infty(M), D_{\Lambda'})$ define a metric on $\statespace$ for some $\Lambda' >> \Lambda$.

For the purposes of the following sections, we will regard $\trlambdist$ as the natural metric on $\statespace$ and leave open the question of Gromov-Hausdorff convergence.

\begin{remark}
Our approach of using $\embedding$ to quantify state localization has the advantage of computational feasibility, at the expense of requiring explicit knowledge of the $2n$ elements $P_\Lambda \embedding_i P_\Lambda$ and $P_\Lambda \embedding_i^2 P_\Lambda$ of $C^\infty(M)_\Lambda$.
The computationally more challenging but perhaps conceptually more satisfying approach of defining the dispersion via $\sup_{a \in C^\infty(M)_\Lambda} \left\{ \left| \omega_v(a^* a) - |\omega_v(a)|^2 \right| \mid \norm{[D_\Lambda, a]} \leq 1 \right\}$, however, has an additional drawback: it requires knowledge of the pairs $(P_\Lambda a P_\Lambda, P_\Lambda a^*a P_\Lambda)$ for all $a \in C^\infty(M)_\Lambda$.
\end{remark}

\section{The \algo algorithm: associating a finite metric space}
\label{sec:algorithm}

Once a set of localized vector states is found, the Connes (Kantorovich-Rubinstein) distance between them will serve as an estimate for the geodesic distance between the points in $M$ near which they are concentrated.
Keeping in mind the discussion of Section~\ref{sec:metriconstates}, we will regard the truncated metric of Equation~(\ref{eq:truncated-distance}) as the natural metric on $\statespace$.

Localized vector states can be found by minimizing the dispersion functional in $H$.
Apart from the comparison of the metrics~(\ref{eq:pullback-distance}) and (\ref{eq:truncated-distance}), nonzero dispersion induces a distortion of estimated distances (see section \ref{sec:resolution}, below).
Therefore, the dispersion supplies a lower bound on the Gromov-Hausdorff distance between any graph of localized states and the manifold $M$.
Computationally speaking, then, it would be desirable to minimize the number of states (and, hence, computational resources) required to approach this bound.

The main other factor, besides correctness of distances, that influences the Gromov-Hausdorff distance is the density (in the Hausdorff sense) of our set of points inside $M$.
Optimally, therefore, the states would be equidistributed on $M$.

In order to construct a potential whose minima are both localized and roughly (that is, under the map $\embedding$) equidistributed, we add an electrostatic repulsion term to the dispersion.
Given a set $V$ of states, the next state is then generated as the minimum of the energy functional
\begin{align}
\label{eq:energyfun}
  e(v; V) \eqdef -\disp(v)^{-1} + g_e \sum_{w \in V} \left(\sum_i \left(\langle v, \embedding_i v \rangle - \langle w, \embedding_i w \rangle \right)^2\right)^{-1}.
\end{align}
The value of the coupling constant $g_e$ should ideally be sufficiently large to overcome local variation in minimal dispersion but is otherwise not expected to influence the generated states much -- this is consistent with our observations for $M = S^1, S^2$.

\subsection{The \algo algorithm}

Using the functional \eqref{eq:energyfun} we propose the following algorithm to construct states and thus a finite metric space $M_\Lambda$ that models the metric information about $M$ contained at cutoff $\Lambda$.

As preparation, we must estimate the number $N$ of states to generate.
\begin{itemize}
\item
  Estimate $\vol M_\Lambda$ and $\dim M_\Lambda$, e.g. using the asymptotic formulas of \cite{Stern_2018}.
\item
  Estimate the Euclidean dispersion $\disp_0 = E_\nu \left[ \| X \|^{2} \right]$ under the multivariate normal distribution $\nu$ of covariance matrix $2 \Lambda^{-2} \log \Lambda \id$ on $\R^{\dim M_\Lambda}$.
\item
  Set $N = \vol M_\Lambda / (\vol(B_{\dim M_\Lambda}) \disp_0^{\dim M_\Lambda / 2} )$, where $B_{\dim M_\Lambda}$ is the Euclidean unit ball of dimension $\dim M_\Lambda$.
\end{itemize}
For cases where $\embedding$ is a Riemannian embedding of $M$, any $g_e$ will suffice to lead to equidistributed states in $M$, while for $g_e=0$ the states generated numerically would mostly lie very close together.
However in the cases where $\embedding$ is far from Riemannian, we need to chose $g_e$ to be sufficiently large to overcome local variations in minimal dispersion, and assume this to mean that $-\alpha^2 \disp_0^{-1} + g_e \alpha^2 \disp_0^{-1} \geq -\beta^2 \disp_0^{-1}$, where $\alpha$ and $\beta$ are the optimal local Lipschitz constants of $\embedding$ and $\embedding^{-1}$, respectively.
This ensures that states in regions of $M$ where the dispersion is over-reported (due to stretching by $\embedding$) will be generated once the regions where the dispersion is under-reported are saturated with states, instead of being skipped.

Then, simply generate $N$ states by minimizing the iterative energy functional and calculate the Connes distance between them:
\begin{algorithmic}[1]
\While{$|V| \leq N$}
   \State Find a vector $w$ (locally) minimizing $e(w; V)$. \label{algstep:states}
   \State Append $w$ to $V$.
   \For{$v \in V$,}
   \State Set $d(v, w) = \max \{ |\langle v, a v \rangle - \langle w, a w \rangle| \colon \norm{[D,a]} \leq 1 \}$. \label{algstep:connes}
   \EndFor
\EndWhile
\end{algorithmic}

The algorithm, including the distance calculation and the examples $S^1$ and $S^2$, has been implemented in Python and is publicly available at \cite{graphmaker_repo}.

\subsection{Implementation: calculating the metric on $\statespace$}
When $v, w \in H_\Lambda$, the distance between the vector states $\bra{v} \cdot \ket{v}$ and $\bra{w} \cdot \ket{w}$ of the algebra $A = C^\infty(M)$ equals
\[
  \max _{a \in A_\Lambda} \{ |\langle v, a v \rangle - \langle w, a w \rangle| \colon \norm{[D_\Lambda,a]} \leq 1 \},
\]
as in discussed in section \ref{sec:metriconstates}.

The functional $a \mapsto \bra{v} a \ket{v} - \bra{w} a \ket{w}$ is linear and the space $\{ a  \in A_\Lambda \mid \norm{[D_\Lambda, a]} \leq 1 \}$ is convex, which ensures that computing the minimum is computationally feasible.

Indeed, if $a_0, \dotsc, a_n$ is a basis for $(A_\Lambda)_{\text{sa}}$, we can reformulate the problem as:
\begin{problem}
  Minimize $\sum_i c_i \left( \langle v, a_i v \rangle - \langle w, a_i w \rangle \right)$ over $c \in \R^{n+1}$, subject to the constraint
  \begin{align*}
    \begin{bmatrix}
      I & \sum_i c_i [D_\Lambda,a_i] \\
      \sum_i c_i [D_\Lambda,a_i]^* & I
    \end{bmatrix}
                                     > 0
  \end{align*}
\end{problem}
With the constraints formulated as a linear matrix inequality, we have put the problem in a form directly amenable to techniques from semi-definite programming.
A reasonably effective algorithm, given the scale of the problem, is then provided by the Splitting Cone Solver of \cite{o2016conic}.
\subsection{Complexity and the dimension of $C^\infty(M)_\Lambda$}
\label{sec:dim-of-pap}
Step \ref{algstep:states} of the \algo algorithm amounts to finding a local minimum of a quadratic function under quadratic constraints in a vector space of dimension $\dim H_\Lambda$, which can be done in $O(\dim H_\Lambda)^2$, e.g. with the BFGS algorithm.

The problem in step \ref{algstep:connes} is convex, of dimension $\dim C_\Lambda(M)$.
This factor is what limits the computational feasibility of high $\Lambda$ in our experiments, so it would be informative to analyze the scaling of $\dim C_\Lambda(M)$ with $\Lambda$.

As a simple example, one can represent the generator $e^{i \theta}$ of $C^\infty(S^1)$ as the shift operator on $H = l^2$, with basis indexed by $\Z$, where the corresponding Dirac operator acts diagonally as $D e_n = n e_n$. It is then easy to see that the dimension of $C^\infty(S^1)_\Lambda$ is equal to $\dim H_\Lambda = 2 \lfloor \Lambda \rfloor + 1$.

For $M = S^2$, if we choose an orthonormal basis $e_{lm}$ of eigenvectors of $D$ and introduce the spherical harmonics $_{0}Y_{lm}$ then we can express $\langle (e_{l_1 m_1} \cdot e_{l_2,m_2}),_{0} Y_{l_3 m_3} \rangle_{L^2(M)}$ in terms of $3j$-symbols. In particular, these vanish unless $||l_1|-|l_2|| \leq l_3 \leq |l_1| + |l_2|$, which tells us that $C^\infty(S^2)_\Lambda$ is spanned by $({}_0Y_{lm})_\Lambda$ for $l \leq 2 \Lambda$ and is thus of dimension bounded by $(2 \Lambda + 1)^2$.

The general situation is not entirely clear. However, our Lemma~\ref{prop:pointsarestates}, as noted there, provides a lower bound of $\Theta(\Lambda^{\dim M})$ on the scaling of $\dim C^\infty(M)_\Lambda$ with $\Lambda$.

\section{Example: $S^2$}
\label{sec:example}

The simplest interesting example of a commutative spectral triple that allows for an isometric embedding in $\R^3$ is probably the sphere $S^2$.
This section will cover the application of the \algo algorithm to truncations of $(C^\infty(S^2), L^2(S^2, \spinorbundle_{S^2}), D_{S^2})$, and thereby illustrate (and test the optimality of) the analytic results of Section~\ref{sec:localized}.

\subsection{Implementation}
The main ingredients are the vector space $C^\infty_{\Lambda}(S^2)$, the spectrum of $D_{S^2}$, the element $\embedding$, and their representation on $L^2(S^2, \spinorbundle_{S^2})_\Lambda$.

The vector space $C^\infty_{\Lambda}(S^2)$ is spanned by the spherical harmonic functions $Y_{lm}$ up to $l = 2 \Lambda$, as in section \ref{sec:dim-of-pap}.
An eigenbasis $e_{lm}$ of $D$ can be expressed in terms of the spin-weighted spherical harmonics $_{s}Y_{lm}$, with $s = \pm \frac12$, as discussed e.g. in \cite{gracia2013elements}, section 9.A.
The matrix coefficients of the representation of $C^\infty_{\Lambda}(S^2)$ can then be expressed in terms of triple integrals of spin-weighted spherical harmonics.
Note that a brute-force approach of calculating the inner products $\langle e_{lm}, Y_{l'm'} e_{l''m''} \rangle$ in order to obviate knowledge of the representation-theoretic machinery attached to $S^2$ would have been possible, however it would have introduced the additional complexity of calculating $(\dim H_{\Lambda})^2 \cdot \dim C^\infty_{\Lambda}(S^2) \cdot \rank \spinorbundle$ integrals numerically.

The element $\embedding$ is just the idempotent associated to the Bott projection: $\embedding = \begin{pmatrix} z && x - i y \\ x + i y && - z \end{pmatrix}$, where $x, y, z$ are the standard coordinates on the embedding $S^2 \hookrightarrow \R^3$.
Note that this embedding $\embedding$ is isometric, although that is \emph{not} necessary for the algorithm or the theory in Section~\ref{sec:localized} to work.

The source code to this implementation is publicly available as part of the full Python implementation of the algorithm at \cite{graphmaker_repo}.

\subsection{The localized state densities}

Because the measures associated to states in $\statespace$ are of the form $(v, v) \vol_M$, with $v$ in the finite-dimensional vector space $P_\Lambda H$, one can easily plot the corresponding function $(v, v)$ on $M$.
This allows us to test them, by simply plotting the corresponding fiberwise inner product of the spinor spherical harmonics in the continuum.
We can then compare these with the numerical states generated through the \algo algorithm for different  $\Lambda$.
The expectation is that the numerical states will be comparable to the states obtained through $\statemap$ but will be slightly less localized.

Figure~\ref{fig:LocStates} shows plots for $\statemap(v_x)$, for fixed $v_x \in \spinorbundle$ is fixed, and plots of numerical states for $\Lambda=4, 10$.
It is evident that the states are indeed peaked neatly near $x$, in both cases.
We thus find that the states are well localized and become more localized the larger the cutoff is.

\begin{figure}
\subfloat[][Analytic state for $\Lambda=4$]{\includegraphics[width=0.5\textwidth]{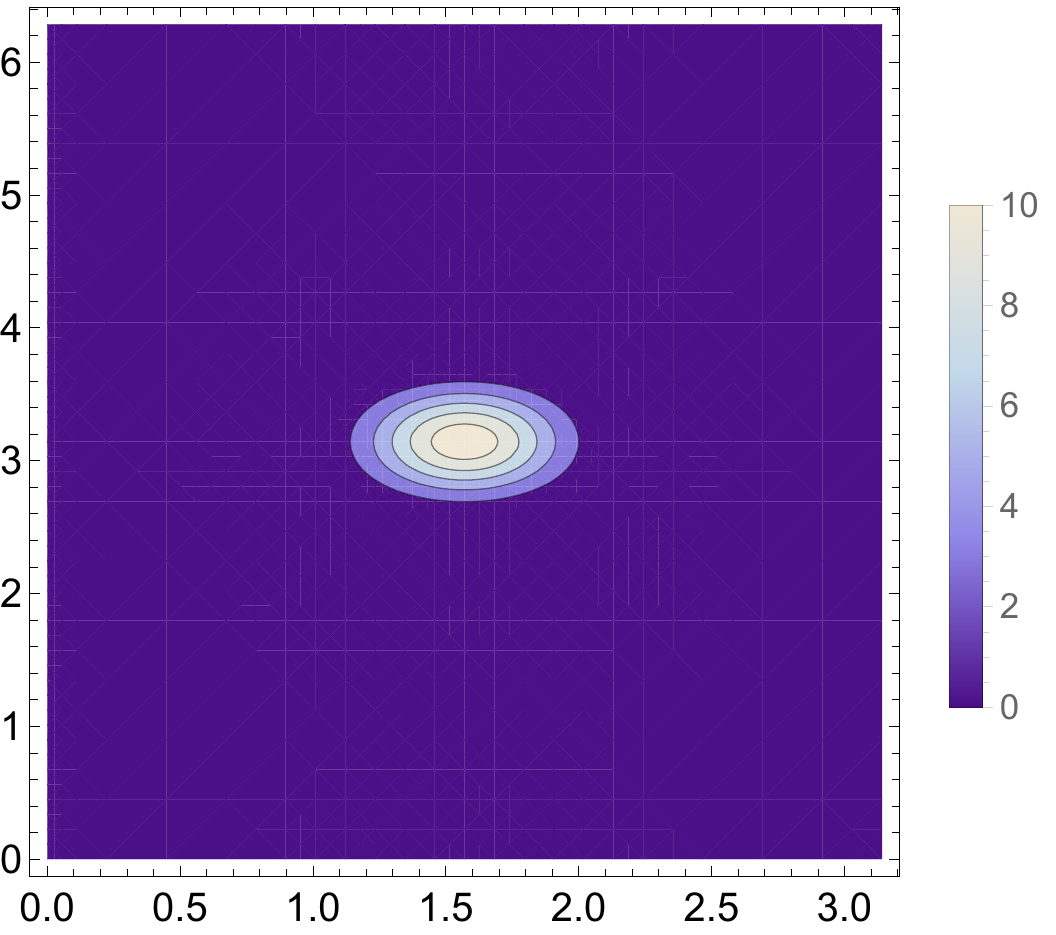}}
\subfloat[][Analytic state for $\Lambda=10$]{\includegraphics[width=0.5\textwidth]{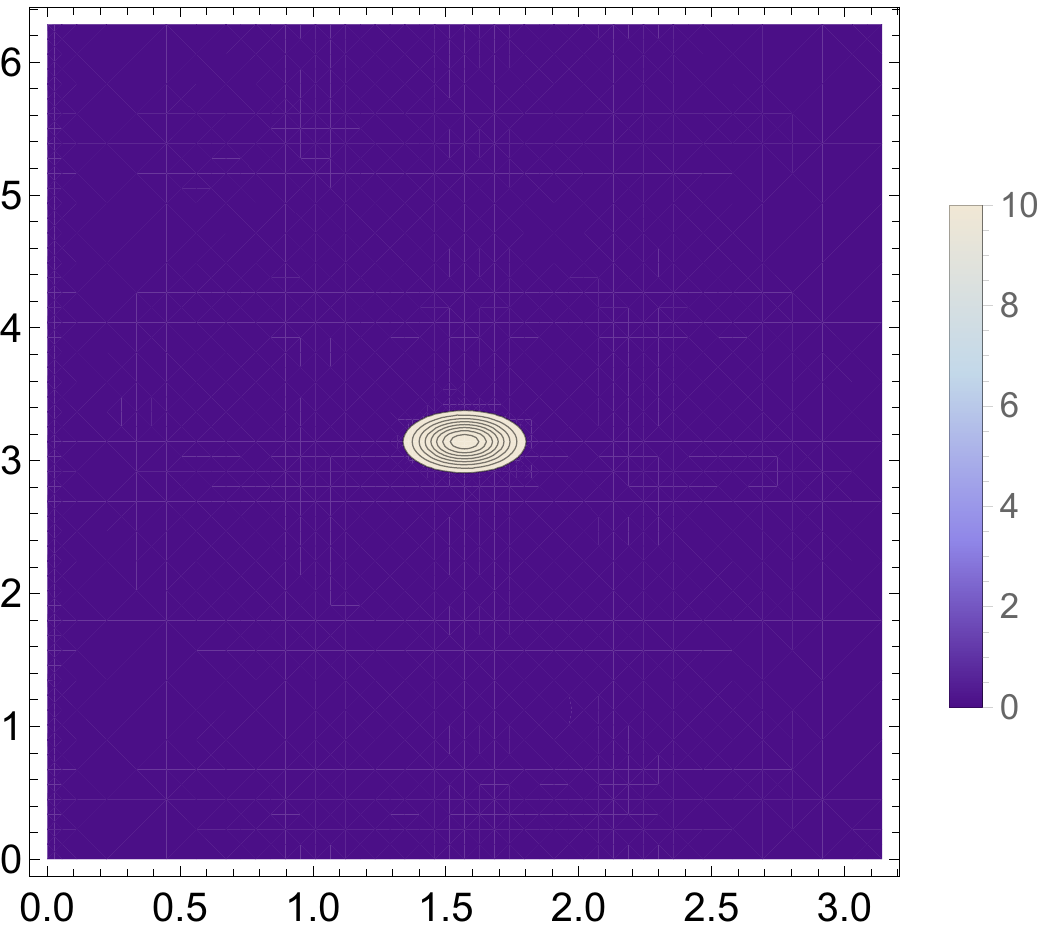}}

\subfloat[][Numerical state for $\Lambda=4$]{\includegraphics[width=0.5\textwidth]{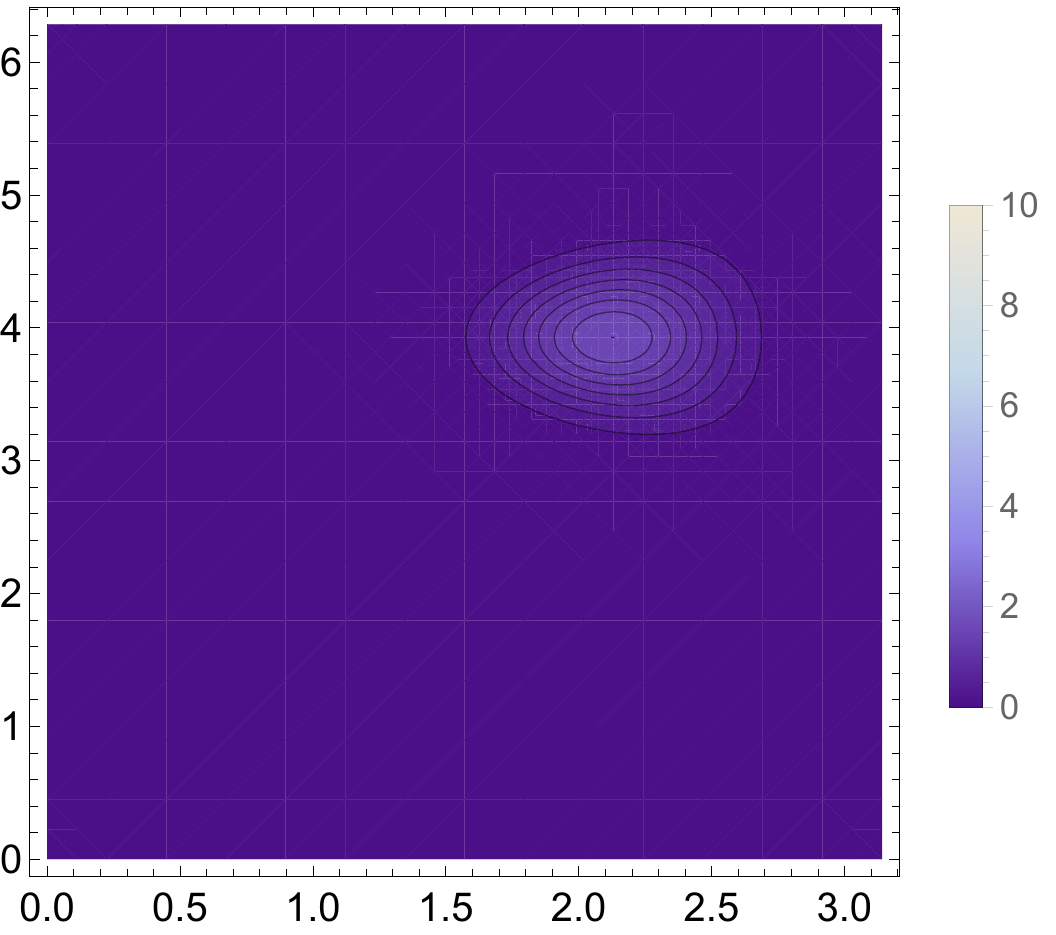}}
\subfloat[][Numerical state for $\Lambda=10$]{\includegraphics[width=0.5\textwidth]{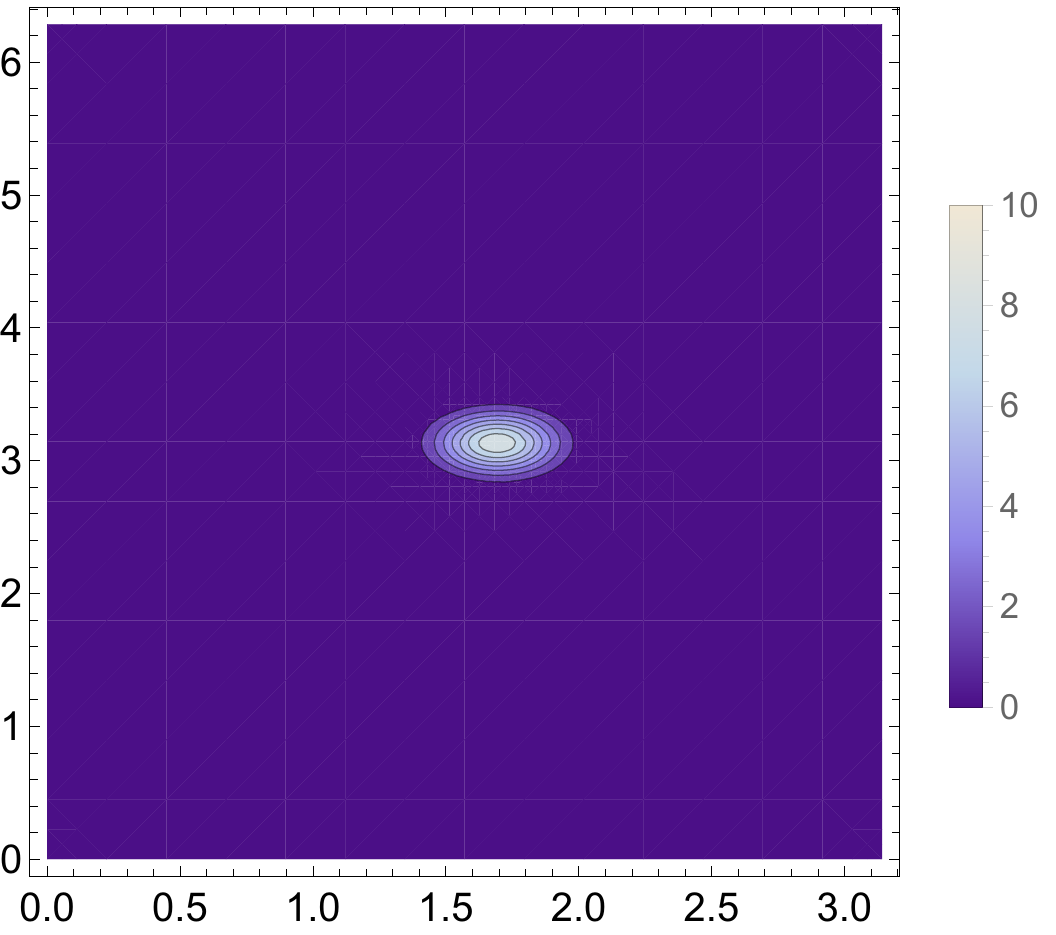}}
  \caption{Plot of analytic and truncated localized states.}
  \label{fig:LocStates}
\end{figure}

Other than this qualitative comparison we also have analytic control.
Lemma~\ref{prop:statesarepoints} gives the functional form of the dispersion as a function of the cut-off $\Lambda$ as $\log{\Lambda}/\Lambda^2$.
We can check this relation explicitly by plotting the size of the dispersion against the cutoff value, as done in Figure~\ref{fig:dispersion} for the cutoff up to $\Lambda=16$.

\begin{figure}
  \centering
\includegraphics[width=0.8\textwidth]{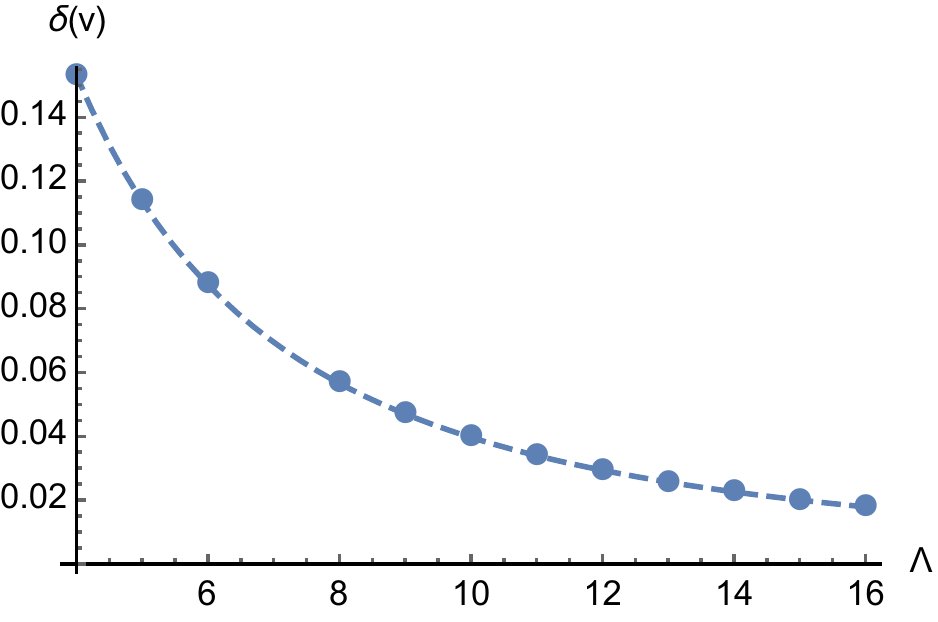}
\caption{Plot of the dispersion of states versus the value of the cutoff for the states. The dashed line is a fit of the analytic result that the dispersion should scale like $\log{\Lambda}/\Lambda^2$.}
\label{fig:dispersion}
\end{figure}

\subsection{Distribution of states over the sphere}

Plotting several states simultaneously allows us to show how the repulsion term distributes them over the sphere. Figure~\ref{fig:states_on_sphere} shows $17$ states for $\Lambda=11$.
The distribution of states in Figure~\ref{fig:states_on_sphere} has some inhomogeneities, some gaps between states are very large.
This is because we only generated 17 states instead of the $110$ we would expect to generate in the \algo algorithm.
Restricting the number of states reduced computation time, and allowed for a clearer visualization of the independent states.
In the right hand Figure we see the states as densities on the sphere, while the left hand plot shows the densities in the $\theta, \phi$ plane.

\begin{figure}
\subfloat[][States in $\theta - \phi$ plane]{\includegraphics[width=0.5\textwidth]{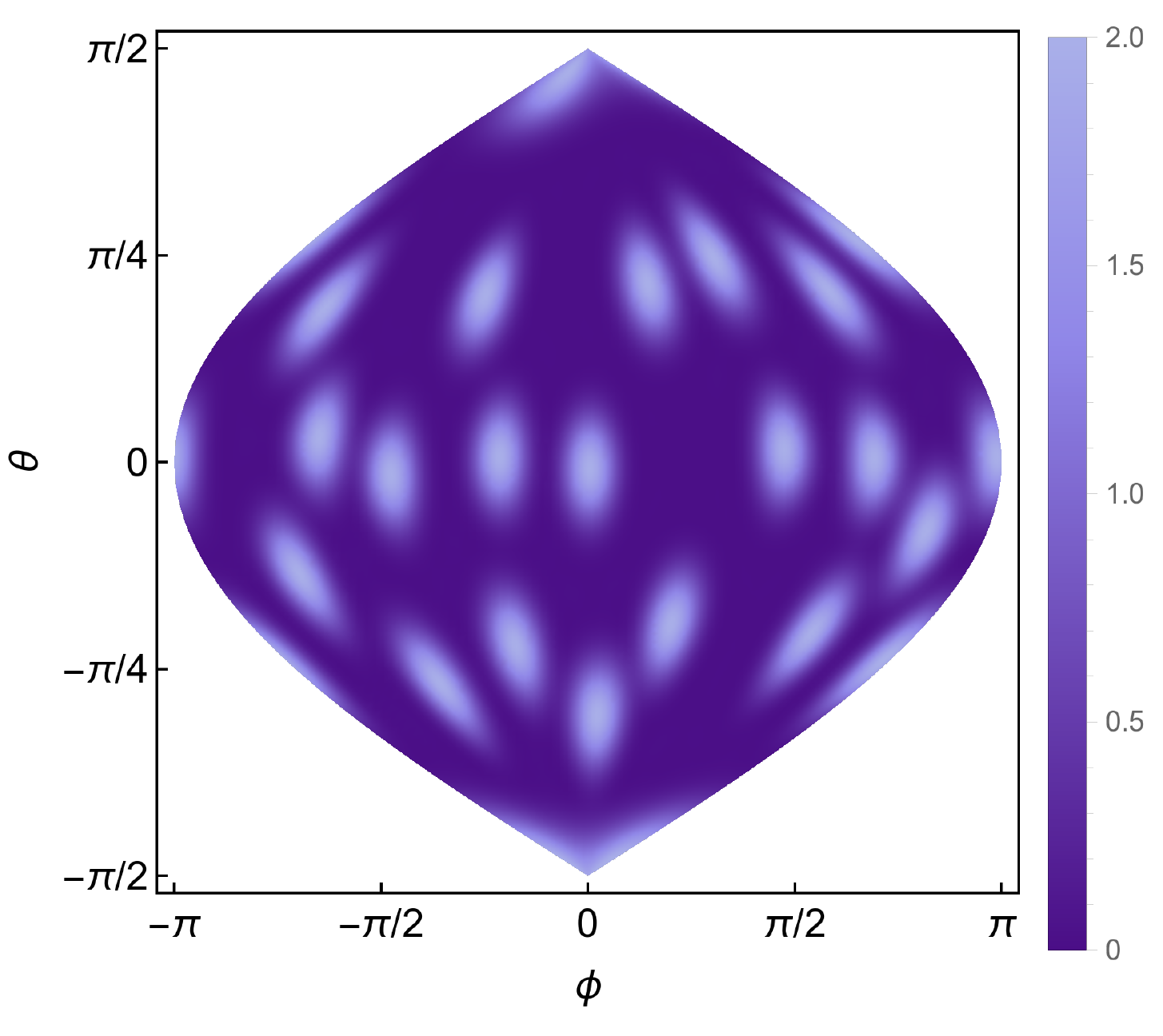}}
\subfloat[][States on the sphere]{\includegraphics[width=0.5\textwidth]{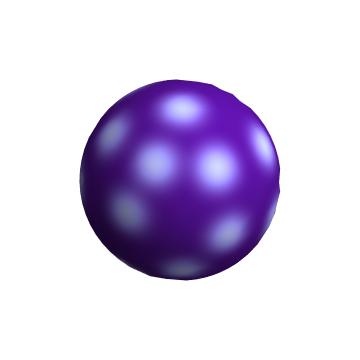}}
\caption{Localized states on the sphere, the left hand image shows the states projected on the two dimensional plane using a sinusoidal projection, while the right hand image shows the states on the sphere.}
  \label{fig:states_on_sphere}
\end{figure}

To test how the repulsive potential acts we can generate states on the sphere and just plot the coordinates for their center of mass associated with the embedding maps $\embedding_i$.
We show this in Figure~\ref{fig:repulsion} for a maximal eigenvalue of $\Lambda=10$, it is clear that without potential all states generated cluster at one point, while even a weak repulsive potential leads to points that are evenly distributed.

\begin{figure}
  \subfloat[][No repulsive potential]{\includegraphics[width=0.5\textwidth]{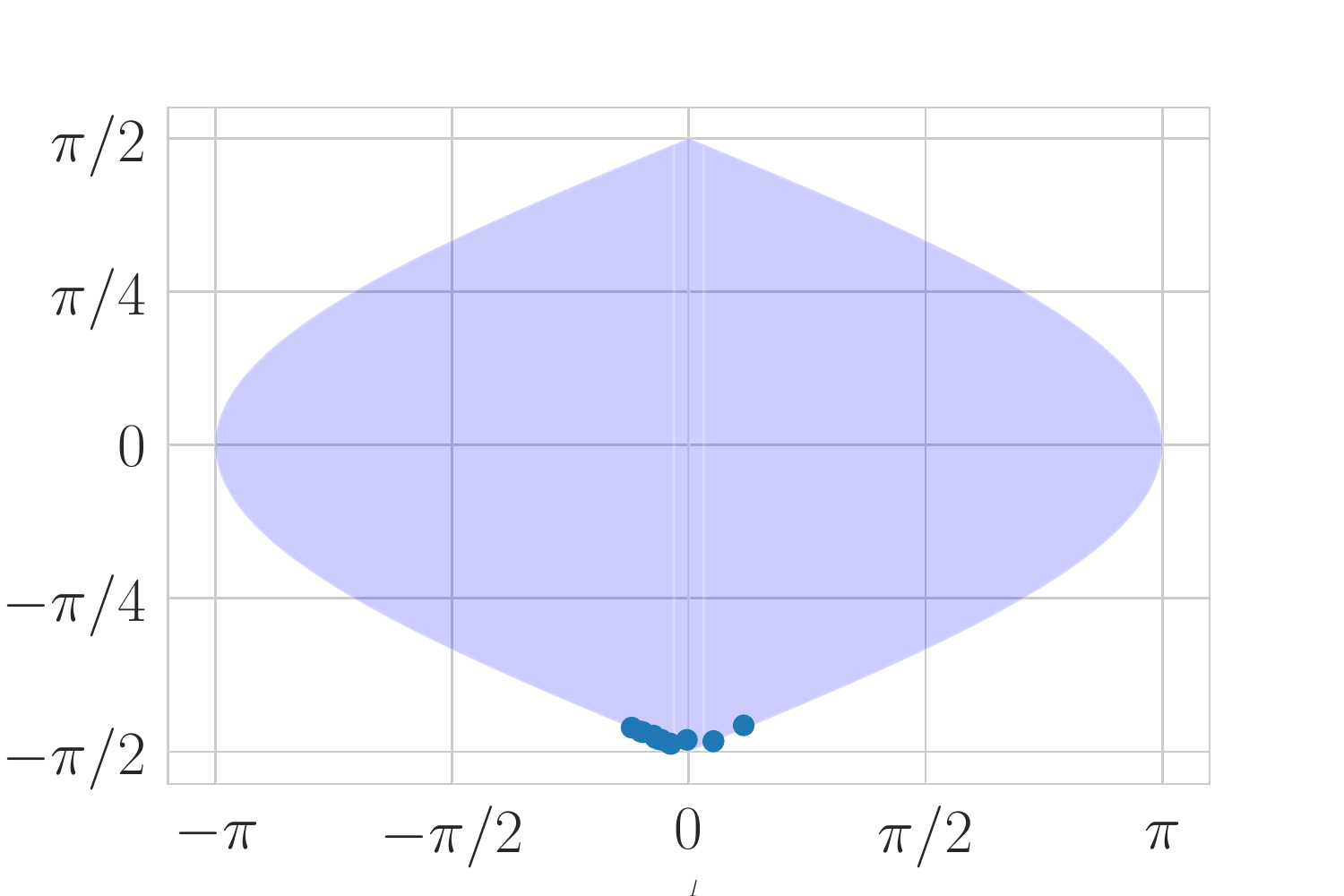}}
  \subfloat[][Repulsive potential with coupling $0.001$]{\includegraphics[width=0.5\textwidth]{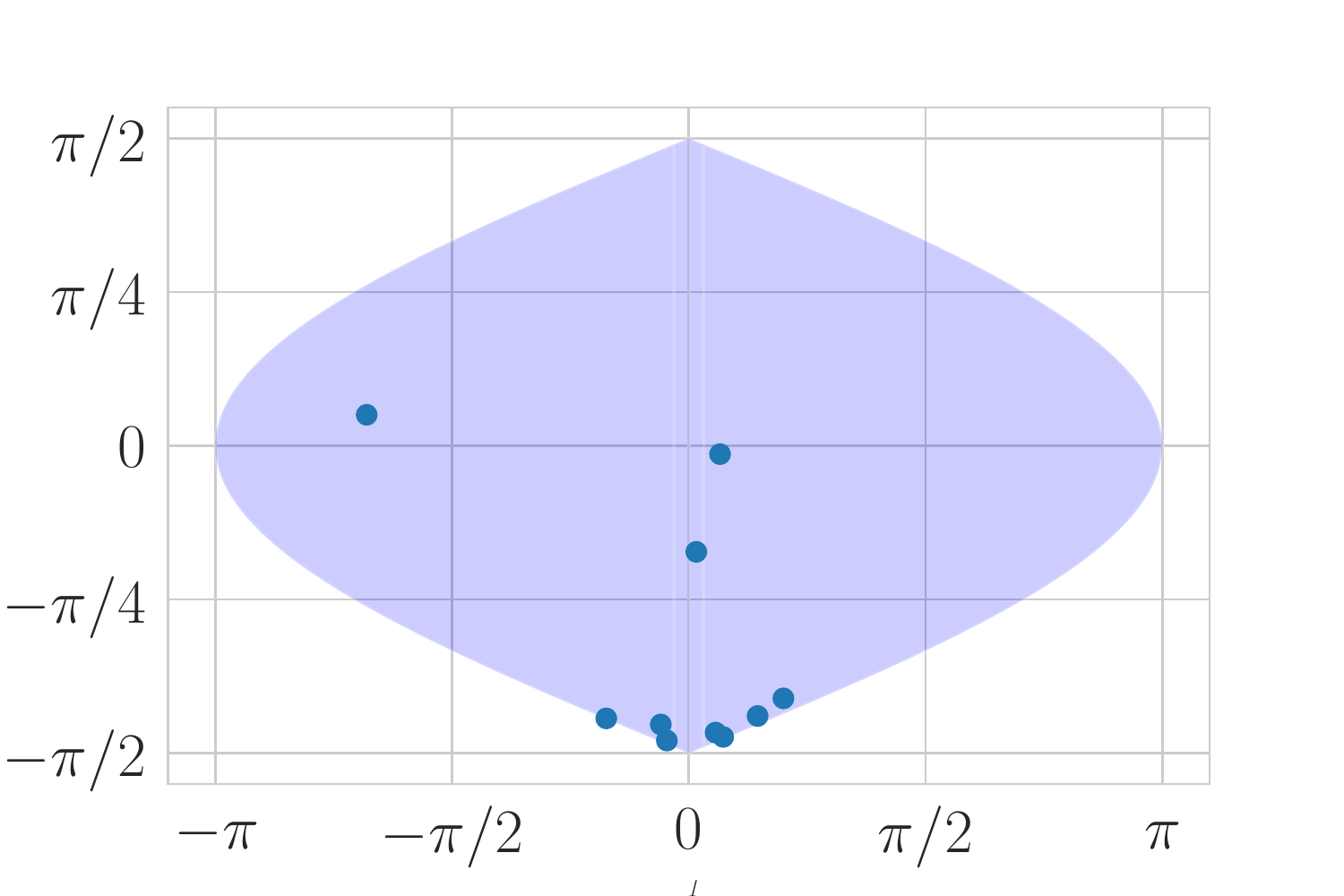}}

  \subfloat[][Repulsive potential with coupling $0.1$]{\includegraphics[width=0.5\textwidth]{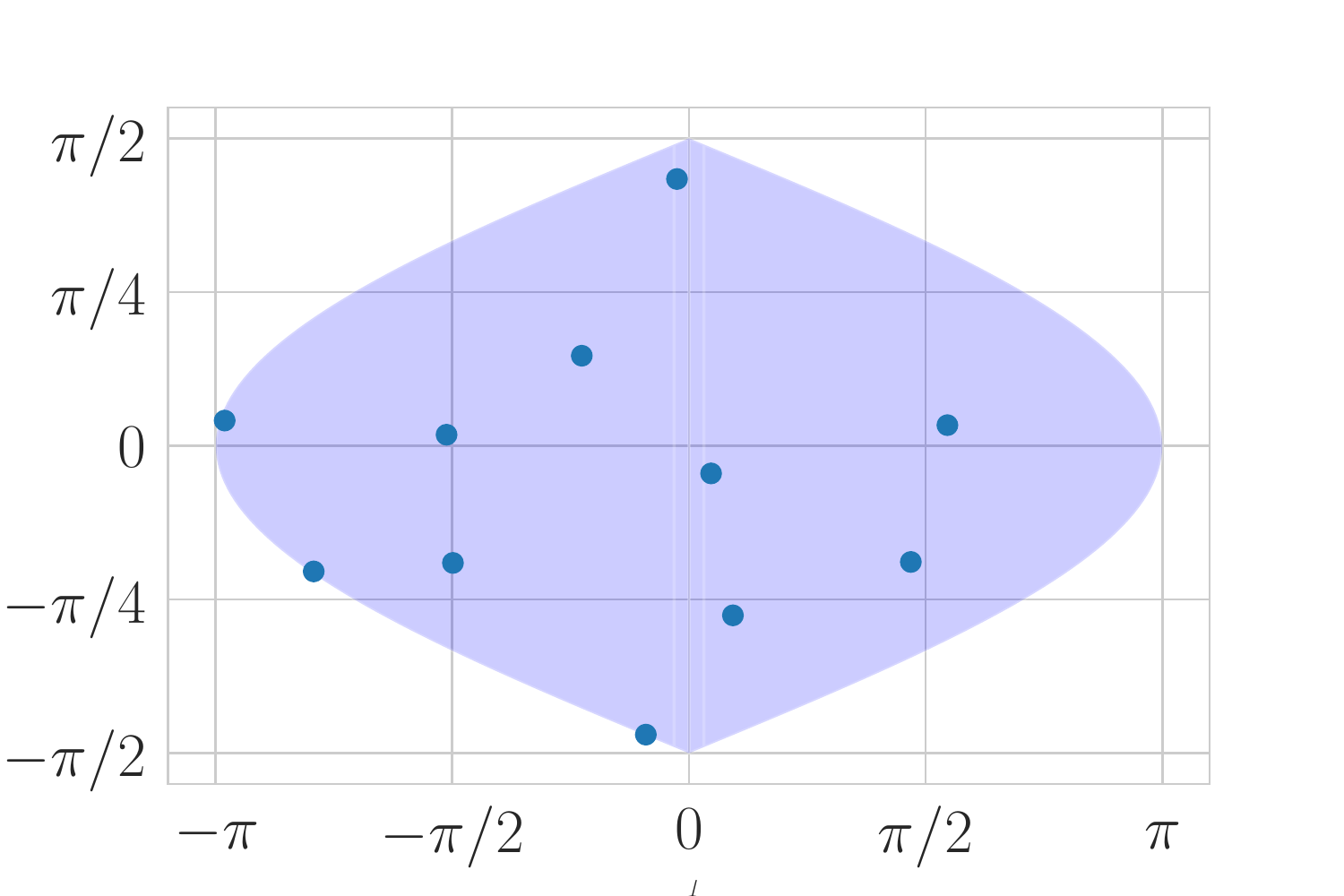}}
  \subfloat[][Repulsive potential with coupling $100$]{\includegraphics[width=0.5\textwidth]{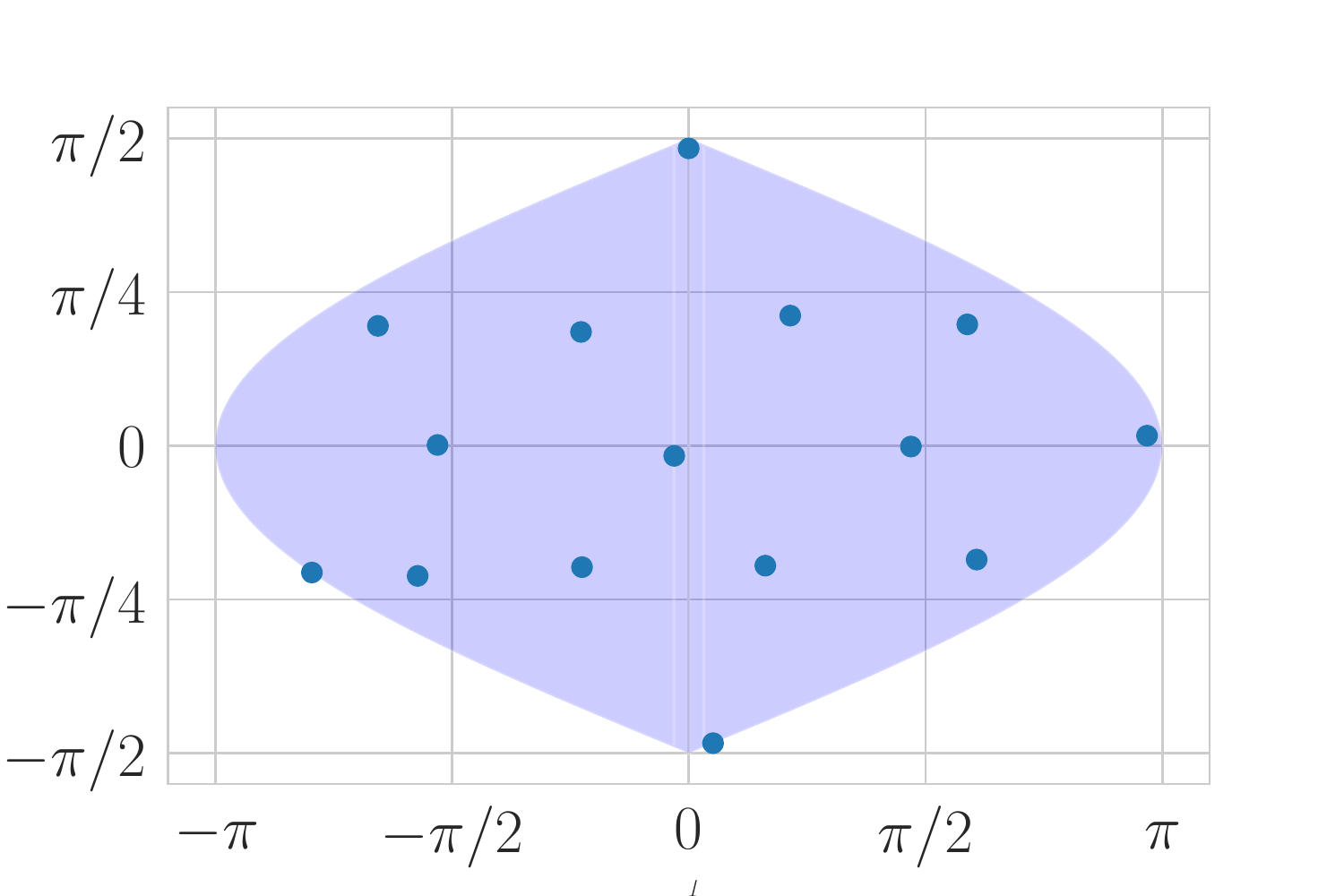}}
  \caption{This shows how the states are distributed dependent on the repulsive potential.
  We can see that even a weak repulsive potential suffices to lead to well distributed points.
  This figure shows the point distribution in the Sinusoidal projection, flattening the sphere onto the plane.}
  \label{fig:repulsion}
\end{figure}

\subsection{Error analysis}
\label{sec:sphere_graph_error_analysis}

If a measure $\mu$ on $S^2$ is reasonably localized, so that $x \eqdef E_{\mu}\left[ \embedding(X) \right] \neq 0$, then it possesses a unique $\embedding$-barycenter $p$ given by the projection of $x$ onto the sphere.
The Euclidean distance between $x$ and $\embedding(p)$ is then given by $\sqrt{1 - \norm{x}^2}$.

Figure~\ref{fig:truncated-distance-errors}(a) shows how closely the geodesic distance between barycenters is approximated by the truncated Connes distance $\trlambdist$.
The monotone scaling of the error reflects the fact that antipodal, imperfectly localized measures are significantly closer in Wasserstein distance than their barycenters are, due to the presence of the cut locus. 

Interestingly, the error is strictly positive, so that $\trlambdist(\mu_1, \mu_2)$ turns out to be -- for the states considered -- a \emph{better} approximation to $\contdist(p_1, p_2)$ than $W_1(\mu_1, \mu_2) = \lambdist(\mu_1, \mu_2)$ itself.
In particular, as long as the error is positive, the convergence of $\lambdist(\mu_1 , \mu_2)$ to $\contdist(p_1, p_2)$ as the dispersions fall implies convergence of $\trlambdist$ to $\lambdist$ as well.
Whether this points to special behaviour of the (truncated) Connes distance between \emph{localized} elements of $\statespace$ remains to be seen.

For measures $\mu_1, \mu_2$ on $S^2$, the analysis of Section~\ref{sec:resolution} shows that $|\contdist(p_1, p_2) - W_1(\mu_1, \mu_2)| \leq W_2(p_1, \mu_1) + W_2(p_2, \mu_2)$, and moreover the latter satisfies $W_2(p, \mu)^2 \leq \beta^2 E_{\mu} \left[ \norm{ \embedding(p) - \embedding(X) }^2 \right]$, where $\beta = \pi/2$ is the Lipschitz constant of $\embedding$.
Therefore, we have an estimate $|\contdist(p_1, p_2) - W_1(\mu_1, \mu_2)|^2 \leq  \pi^2 / 4 \left(\sum_i \disp(\mu_i) + (1 - \norm{x_i}^2) \right)$.
Since the other terms in this inequality can readily be calculated, it provides us with a theoretical lower (upper) bound on $W_1(\mu_1, \mu_2)$.
Figure~\ref{fig:truncated-distance-errors}(b) shows this lower bound, and the upper bound provided by $\trlambdist$, for $W_1(\mu_1, \mu_2)$, with $\contdist(p_1, p_2)$ shown for reference.

\begin{figure} \label{fig:truncated-distance-errors}
  \subfloat[][Absolute error]{ \includegraphics[width=0.5\textwidth]{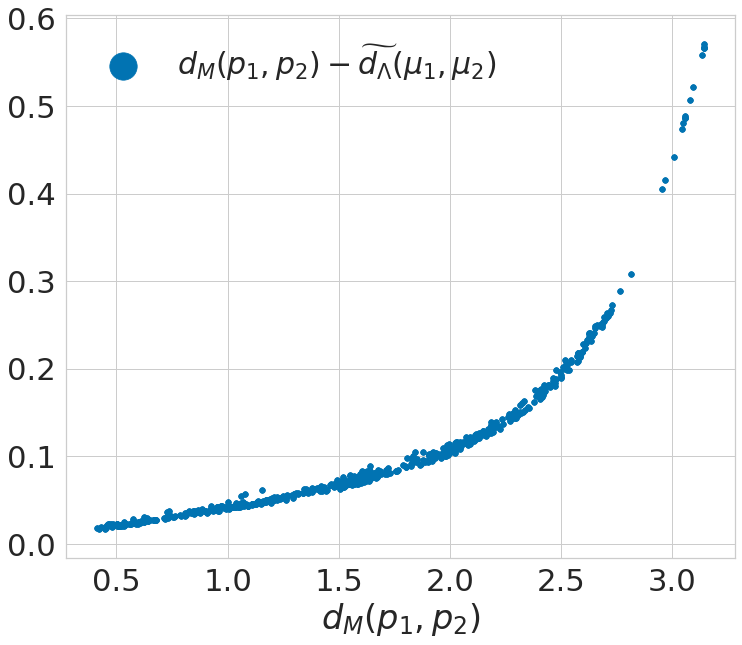}
  }
  \subfloat[][Bounds for $\lambdist(\mu_1, \mu_2)$]{
    \includegraphics[width=0.5\textwidth]{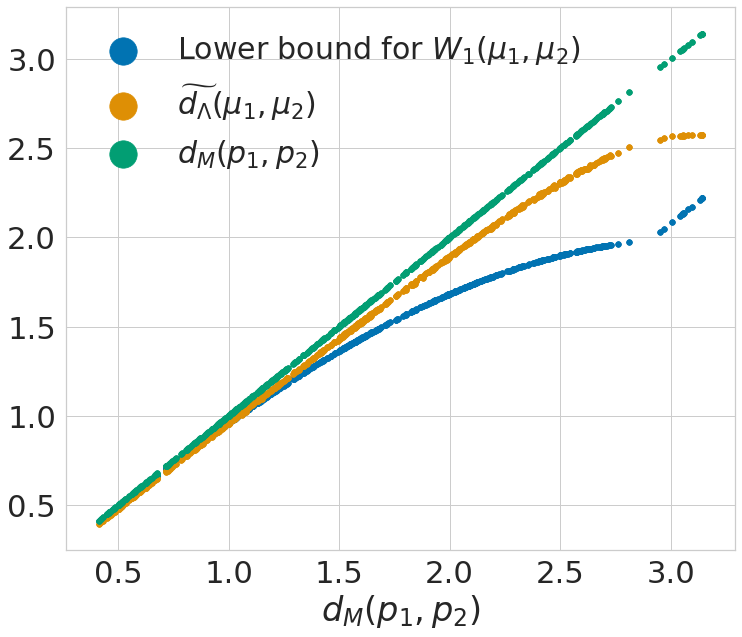}
    }
    \caption{Distance errors and bounds for pairs of localized states at $\Lambda = 5$.}
\end{figure}

\section{Embedding a distance graph in $\R^n$}
\label{sec:embedding_in_Rn}
Let $M$ be a Riemannian manifold, and assume that $M$ embeds isometrically into $\R^n$.
Now take a finite set $V$ of points in $M$ and their geodesic distances $d(\cdot, \cdot)|_{V \times V}$, e.g. by generating states as above and calculating their distances using~\eqref{eq:distance}.
Optimally, we would ask for a way to embed $V$ in $\R^n$ such that its image under this embedding equals its image under some Riemannian isometry $M \to \R^n$.

Of course, without knowledge of $M$ such a problem is unsolvable for any given $V$.
Instead, we hope for our embedding procedure to satisfy such a property asymptotically, i.e. that for sequences of $V$ of increasing density their embeddings converge to an embedding of $M$, under some suitable notions of density and convergence.
This is an open problem, and since our primary purpose at this point is one of visualization, we will only take it as a guiding principle.

\subsection{Stress and local isometry of embeddings}
The field of optimal graph embedding is well-established and provides many approaches to questions similar to the above.
A particular model of interest is \emph{metric multidimensional scaling}\footnote{See e.g. \cite{borg_modern_2005}}, where one looks for an embedding $X: V \to \R^n$ that minimizes the stress function,
\[
\sigma(X) = \frac{\sum_{p \neq q \in V} w(p,q) (d(p, q) - ||X(p) - X(q)||)^2}{\sum_{p \neq q \in V} w(p, q)},
\]
where $w$ is a positive weight function: this is just a weighted version of the second Gromov-Wasserstein distance between $V \subset M$ and $X(V) \subset \R^n$.

Because our $M$ is not assumed to be Euclidean, the usual choice $w = 1$ would be quite unnatural here.
In particular, an isometric embedding of $M$ in the Riemannian sense would not necessarily have minimal stress, because the model instead asks for isometry in the sense of maps of \emph{metric spaces}, not Riemannian manifolds.
Since all tangent space information is lost when discretizing like this, the Riemannian notion of isometry does not translate immediately and we must replace it using a measure of locality.

By the smoothness of an isometric embedding $\phi$ of $M$, the relative defect  $|d(p, q) - ||\phi(p) - \phi(q)||| / d(p, q)$ must converge to $0$ as $p \to q$.
That is to say, as long as we only worry about pairs of points that are \emph{close} in $M$, the stress function above places the correct restriction on $X$ - the further they are apart, the less sense the corresponding contribution to $\sigma$ makes.
This motivates us to pick a positive weight function $w(p, q)$ of $d(p, q)$ that decays monotonically and sufficiently quickly to suppress those lengths that cannot be approximated well by an Euclidean embedding.

For example, imagine two points connected by a shortest geodesic (a great circle arc) of length $l \leq \pi$ on the unit sphere and let that sphere be embedded isometrically in $\R^3$.
In $\R^3$, the shortest geodesic connecting the points is a chord of length $c(l) = 2 \sin (l/2)$.
The defect for small geodesic distances $l$ is thus quite small, being $O(l^3)$.
It reaches its maximum when the points are antipodal, with a relative error of $(\pi - 2) / \pi$.
The weight function should suppress the contribution of the larger distances to the stress $\sigma$, in order to still recognize when an embedding of the distance graph is \emph{locally} isometric.

Let $\phi: M \to \R^n$ be isometric and let $w_k$ be a sequence of weight functions depending on the cardinality $k$ of $V \subset M$. If $w_k(l) = o(1)$ for fixed $l$ and the marginal defect, which is bounded by
\[
  \frac{k \sup_{p, q \in M} w_k(p, q) \left(d(p, q) - \| \phi(p) - \phi(q) \| \right)^2}{\inf_{|V| = k} \sum_{p, q \in V} w_k(p, q)},
\]
 is summable in $k$, we can at least be sure that the stress function $\sigma$ converges to $0$ for embeddings $X = \phi|_{V}$.

The optimal choice of $w$ then depends (at least somewhat) on the geometry of $M$ itself; the curvature $\inf_{\phi: M \to \R^n} \sup \{ |d(p, q) - \|\phi(p) - \phi(q)\| | \mid p, q \in M, d(p, q) \leq \epsilon\}$, as function of $\epsilon$, together with the Hausdorff distance between $V$ and $M$, determines the optimal behaviour of $w$.

\subsection{Implementation}

For $\dim M = 2$, we expect the length of the smallest edges to scale roughly as $k^{- 1/2}$.
For $w_k(l) = \exp \left(- \sqrt{k} l \right)$ the infimum in the denominator of the marginal defect, above, is roughly bounded from below by its value for an equidistributed $V$, which is of order $k^2 \int_0^{\pi} \sin(l) w_k(l) dl \sim k$ as $k \to \infty$. The supremum in its numerator is $O(k^{-3/2})$, so this sequence $w_k$ will do in the narrow sense that it will asymptotically detect when a sequence $\{\phi_k: V_k \to \R^n\}$ corresponds asymptotically to an isometric embedding of $M$, assuming the $V_k$ are roughly equidistributed.

Given the choice of weights, minima of the resulting stress function can be found efficiently using the weighted SMACOF algorithm for stress majorization.
A simple Python implementation of the weighted SMACOF algorithm is part of~\cite{graphmaker_repo}, but for more intensive use we recommend the more efficient FORTRAN version with Python bindings~\cite{smacof_repo}.

\subsection{The $D_c$ operator on the sphere}
\label{sec:dc-intro}

The results of Section~\ref{sec:localized} apply to any Dirac-type commutative spectral triple.
In particular, they apply to perturbations $(C^\infty(M), L^2(M, \spinorbundle), D_{\spinorbundle} + B)$ of a Dirac spectral triple, as long as the perturbation $B$ does not change the principal symbol of $D$.
We will apply the \algo and embedding algorithms both to the sphere and to a perturbation thereof that arose in the companion paper \cite{GlaserSternCCM19}.

All spin manifolds of dimension $\leq 4$ satisfy (the two-sided version of) the higher Heisenberg equation introduced in~\cite{C.C.M:GeometryQuantumBasics}.
The companion paper \cite{GlaserSternCCM19} explores the constraint that existence of solutions to the one-sided higher Heisenberg equation,
\begin{align}\label{eq:CCM}
\frac{1}{n!} \langle Y \underbrace{[Y,D] \dots [Y,D]}_{\text{repeated $n$ times}} \rangle = \gamma,
\end{align}
places on a truncated spectral triple. There we found that \eqref{eq:CCM}, with $Y$ and $\gamma$ obtained from the Dirac spectral triple of $S^2$, is solved by a one-parameter class of operators $\{D_c \mid c \in \R\} \subset \fdspace$, where
\[
  D_c = D_{S^2} + c B.
\]
Here $B$ is a bounded, self-adjoint operator $B=\operatorname{sign}(D) \cos(\pi D_{S^2})$.
This class of solutions does not strictly describe spectral triples, since the pseudo differential operator $D_c$ does not satisfy the first order condition. As discussed there, however, failure of this condition is not detectable by standard methods at the level of truncated spectral triples.

\subsection{Result}
\label{sec:results}

The \algo algorithm returns a metric graph, given an operator system spectral triple $(A, H, D)$ and a designated element $\embedding \in A^n$.
We apply the locally isometric embedding above not only to the example from Section~\ref{sec:example}, but also (tentatively) to the triple $(C^\infty(S^2)_\Lambda, H_\Lambda, D_{c,\Lambda})$ of \cite{GlaserSternCCM19}, in order to investigate the metric properties of the latter.
Here $\Lambda = 5$, corresponding $\dim H_\Lambda = 84$, which leads to $35$ states.

\begin{figure}
\begin{minipage}[t]{0.5\textwidth}
  \subfloat[$\DS$]{\includegraphics[width=0.9\textwidth]{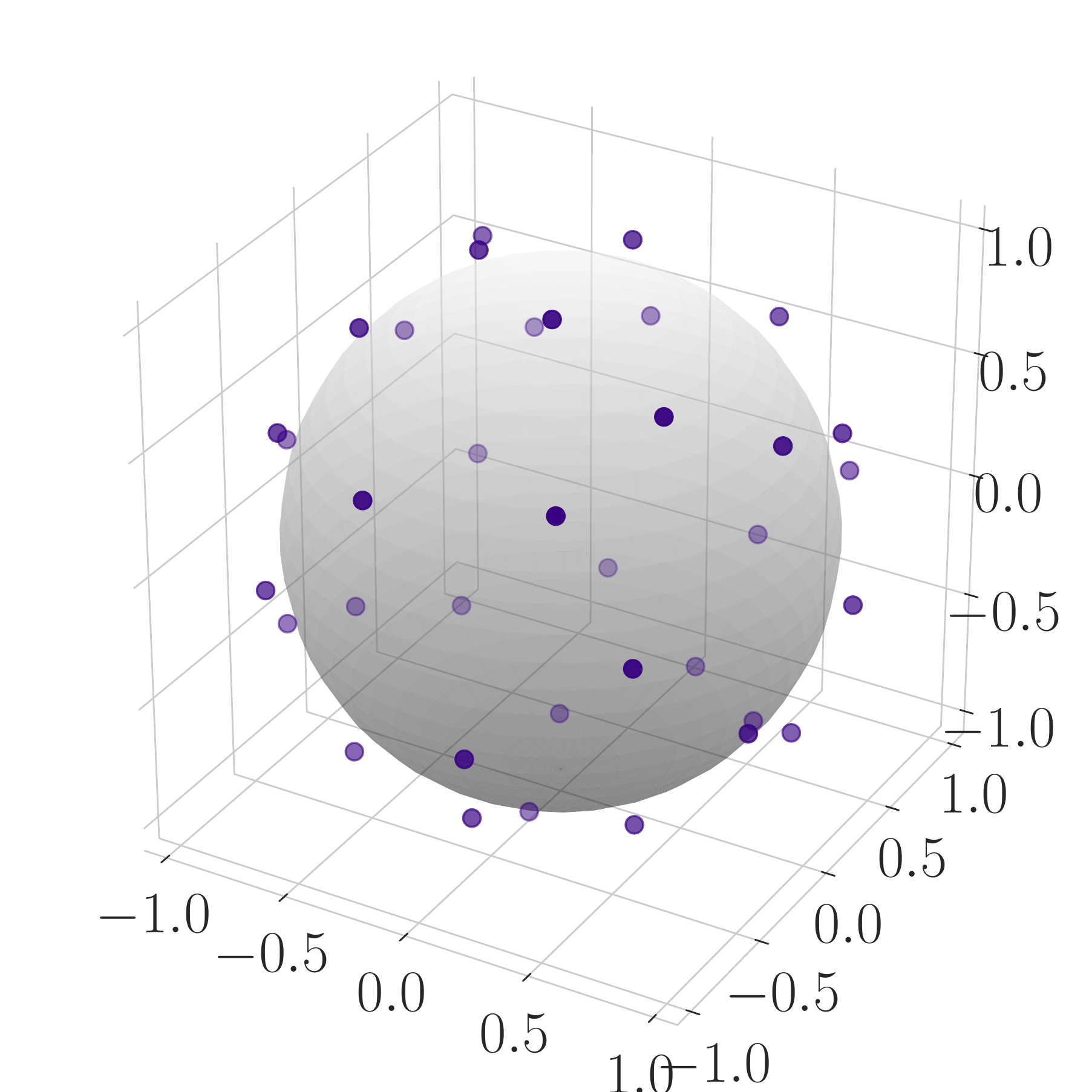}}
\end{minipage}
\begin{minipage}[t]{0.5\textwidth}
  \subfloat[$\DB$]{\includegraphics[width=0.9\textwidth]{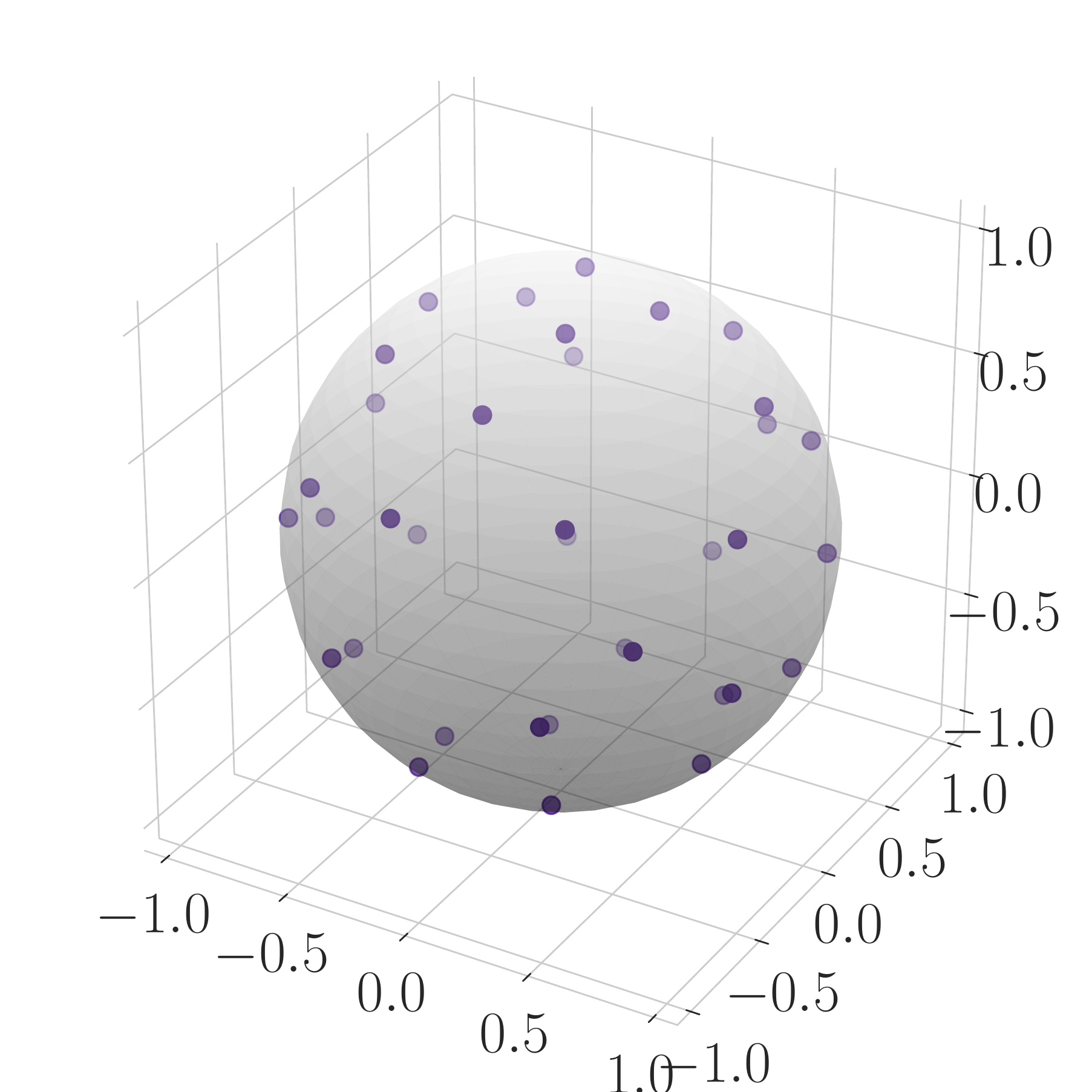}}
\end{minipage}
\caption{Locally almost-isometric embeddings corresponding to $\DS$ and $\DB$, with shaded $S^2$ for reference}
\label{fig:visualisation}
\end{figure}

This leads to the results in Figure~\ref{fig:visualisation}.
There we can see that the embedded points for $\DS$, the left hand plot, lie outside the shaded $S^2$ that is included for reference, while on the other hand the points for $\DB$, in the right hand plot, lie inside the shaded $S^2$.
The transparency of the dots increases with distance to the viewer.
Both embeddings show some deviation from the sphere: for $\DS$, the radii of the embedded points lie in $[1.06,1.12]$, with an average of $1.09$; for $\DB$, in $[0.94,0.98]$ averaging $0.96$.




\section{Final remarks}
\label{sec:conclusion}
The \algo algorithm we introduced in section \ref{sec:algorithm} was designed to reconstruct metric spaces from their truncated commutative (Dirac) spectral triples.
However, the ingredients of the algorithm need not originate as truncations of a commutative spectral triple at all; the steps apply verbatim to arbitrary operator system spectral triples, provided a special 'embedding' element $\embedding$ is given.
Obtaining such $\embedding$ could either be related to the higher Heisenberg equation of \cite{C.C.M:GeometryQuantumBasics}, or, computational resources allowing, be disposed of entirely as discussed in section \ref{sec:local-y-disp}.
This would provide one with the means to construct finite metric spaces associated to an arbitrary noncommutative spectral triple.

It would be interesting to elaborate on this and relate it to quantization and e.g. fuzzy spaces, to get a geometric sense of the relation between a commutative spectral triple and its noncommutative deformations.
This could be particularly useful in connection to more physically inspired explorations of spectral triples and fuzzy spaces, such as~\cite{Barrett_Glaser_2016}.
The ensemble of finite, random spectral triples defined there has shown signs of a phase transition~\cite{Barrett_Glaser_2016,glaserScalingBehaviourRandom2017} and can be characterized through spectral dimension measures~\cite{Barrett:2019aig}, which can be taken as an indication of possibly emergent geometric properties.
The \algo algorithm might then be an interesting tool to further explore some exemplary spectral triples from this class to gather further insights.
It would also be instructive to test how the \algo algorithm works for spectral triples of different topologies, e.g. the non-commutative torus~\cite{Paschke:2006ooy} or a fuzzy torus~\cite{Barrett:2019ize}.

Another possible application in this direction would be to exploit the explicit scale-dependence of the present formalism in order to obtain a better understanding of the gravitational properties of noncommutative approaches (such as \cite{chamseddine1996universal}) to quantum field theory.

\section*{Acknowledgments}
We would like to thank Walter van Suijlekom for extensive discussions and insightful advice.
LG has been funded through grant number M 2577 through the Lise Meitner-Programm of the FWF Austria.
ABS has been funded through FOM Vrij Programma No. 150.

\bibliographystyle{plain}
\bibliography{bib}

\end{document}